\documentclass[journal,onecolumn]{IEEEtran}
\ifCLASSINFOpdf
\else
\fi
\usepackage[latin9]{inputenc}
\usepackage[a4paper]{geometry}
\usepackage{amsmath}
\usepackage{amsthm}
\usepackage{amssymb}
\usepackage{graphicx}
\usepackage{epstopdf}
\usepackage{booktabs}
\usepackage{tikz}
\usepackage[ruled,vlined]{algorithm2e}
\usetikzlibrary{arrows,positioning} 
\tikzset{
    >=stealth',
    punkt/.style={
           rectangle,
           rounded corners,
           draw=black, very thick,
           text width=8.5em,
           minimum height=2em,
           text centered},
     beschr/.style={
           rectangle,
           rounded corners,
           draw=gray, thick,
           text width=8.5em,
           minimum height=2em,
           text centered},
    pil/.style={
           ->,
           thick,
           shorten <=2pt,
           shorten >=2pt,}
}
\newcommand{\D}{{\mathop{}\!\mathrm{d}}} % für Integrad - d's.
\newcommand{\B}{\boldsymbol}
\newcommand{\R}{\mathbb{R}}
\newcommand{\Q}{\mathbb{Q}}
\newcommand{\C}{\mathcal{C}}
\newcommand{\K}{\mathbb{K}}
\newcommand{\N}{\mathbb{N}}

\newcommand{\E}{\mathbb{E}}

\newcommand{\black}{\color{black}}
\newcommand{\one}{ 1 \hspace{-3pt} \mathrm{l}} %

\numberwithin{equation}{section}  
\newtheorem{defn}{Definition}[section]
\newtheorem{exa}[defn]{Example}
\newtheorem{rem}[defn]{Remark}
\newtheorem{thm}[defn]{Theorem}
\newtheorem{prop}[defn]{Proposition}

\usepackage{hyperref}
\usepackage{xcolor}
\hypersetup{
    colorlinks,
    linkcolor={red!50!black},
    citecolor={blue!50!black},
    urlcolor={blue!80!black}
}
\usepackage[caption=false,font=footnotesize]{subfig}

\begin{document}
\title{A deep learning approach to data-driven model-free pricing and to martingale optimal transport}
%\author{Michel Baes\thanks{RiskLab, Department of Mathematics, ETH Zurich, e-mail: $michel.baes@math.ethz.ch$}\and Calypso Herrera\thanks{Department of Mathematics, ETH Zurich, e-mail: $calypso.herrera@math.ethz.ch$}\and Ariel Neufeld\thanks{Division of Mathematical Sciences, NTU Singapore, e-mail: $ariel.neufeld@ntu.edu.sg$}  \and Pierre Ruyssen\thanks{
%		Google Brain, Google Zurich, 
%		e-mail: $pierre.ruyssen@gmail.com$ 

\author{%~\IEEEmembership{Member,~IEEE,}
         %~\IEEEmembership{Fellow,~OSA,}
        Ariel Neufeld %~\IEEEmembership{Fellow,~OSA,}
        and~Julian Sester%,~\IEEEmembership{Life~Fellow,~IEEE}% <-this % stops a space

\thanks{
A. Neufeld 
is with the Division of Mathematical Sciences, 

NTU Singapore.
e-mail: ariel.neufeld@ntu.edu.sg}
\thanks{
	J. Sester 
	is with the Department of Mathematics, 
	
	NUS Singapore.
	e-mail: jul{\_}ses@nus.edu.sg}
%\thanks{Manuscript received January 10, 2022. %; 
%	%revised August 26, 2015.
%}
}


\markboth{A deep learning approach to data-driven model-free pricing and to martingale optimal transport}%
{Shell \MakeLowercase{\textit{et al.}}: Bare Demo of IEEEtran.cls for IEEE Journals}



\maketitle

\begin{abstract}
We introduce a novel and highly tractable supervised learning approach based on neural networks that can be applied for the computation of model-free price bounds of, potentially high-dimensional, financial derivatives and for the determination of optimal hedging strategies attaining these bounds. In particular, our methodology allows to train a single neural network offline and then to use it online for the fast determination of model-free price bounds of a whole class of financial derivatives with current market data. We show the applicability of this approach and highlight its accuracy in several examples involving real market data. Further, we show how a neural network can be trained to solve martingale optimal transport problems involving fixed marginal distributions instead of financial market data.
\end{abstract}
%\begin{IEEEkeywords}
%Model-free price bounds, Deep learning, Data-driven pricing, Martingale optimal transport
%\end{IEEEkeywords}

% \fi
\IEEEpeerreviewmaketitle

%%%%%%%%%%%%%%%%%%%%%%%%%%%%%5
\section{Introduction}
\small{ \IEEEPARstart{F}{inancial} derivatives are financial contracts between the corresponding seller, typically a bank, and a buyer, typically another financial institution or a private person, with a future uncertain payoff depending on another (typically simpler) financial instrument, often a stock, to which we refer as the underlying security. Options are a large class of financial derivatives which allow, but do not oblige the owner of the option to buy or sell the underlying securities involved in the contract. The most common types of traded financial derivatives are call and put options which allow to buy and sell, respectively, the underlying single security at a future maturity at a predetermined price, the so called \emph{strike} of the option. Due to the uncertainty involved in the future cashflow, today's price of the financial derivative is a priori unclear and subject to a high degree of ambiguity. \black The classical paradigm in mathematical finance, which is commonly applied to determine the \emph{fair} value of some financial derivative, consists in capturing the developments of the real underlying market by a sophisticated financial market model\footnote{Examples for sophisticated financial market models include among many others the Heston model (compare~\cite{heston1993closed}) and Dupire's local volatility model (compare~\cite{dupire1994pricing}).}.  \black This model is then calibrated to observable market parameters such as current spot prices, prices of liquid options, interest rates, and dividends, and is thus believed to capture the reality appropriately, see e.g. \cite{schoutens2003perfect} for details of this procedure. However, such an approach evidently involves the uncontrollable risk of having a priori chosen the wrong type of model - this refers to the so called \emph{Knightian uncertainty} (\cite{knight1921}).

To reduce this apparent model risk the research in the area of mathematical finance recently developed a strong interest in the computation of model-independent\footnote{The terms model-free and model-independent are used synonymously in the literature.} and robust price bounds for financial derivatives (compare among many others \cite{beiglbock2013model}, \cite{burzoni2017model}, \cite{cox2011robust}, \cite{davis2014arbitrage}, \cite{dolinsky2014martingale}, \cite{galichon2014stochastic}, \cite{hobson2011skorokhod},  \cite{hobson2012robust}, \cite{hou2018robust}, \cite{neufeld2013superreplication}, and \cite{neufeld2020model}). We speak of model-independent price bounds if realized prices within these bounds exclude any arbitrage opportunities\footnote{ \emph{Arbitrage} refers to a profit that can be realized without taking any risk. Prices of a derivative that allow for arbitrage are considered as not reasonable as the arbitrage profit would be immediately exploited by \emph{arbitrageurs}.\black} under usage of liquid market instruments independent of any model assumptions related to potential underlying stochastic models, whereas robust price bounds refer to the exclusion of model-dependent arbitrage within a range of models that are deemed to be admissible.

We present an  approach enabling the fast and reliable computation of model-independent price bounds of financial derivatives. This approach is mainly based on supervised deep learning (\cite{lecun2015deep}, \cite{schmidhuber2015deep}) and proposes how a deep\footnote{We speak of \emph{deep} neural networks if there are at least $2$ hidden layers involved.} feed-forward neural network\footnote{As a convention we refer to feed-forward simply as neural networks throughout the paper.} can be trained to learn the relationship between observed financial data and associated model-independent price bounds of any potentially high-dimensional financial derivative from an entire parametric class of exotic\footnote{ Every option that is neither a call nor a put option is called \emph{exotic}.} options. The great advantage of the presented methodology is that, in contrast to computational intensive and therefore potentially time-consuming pricing methods which have to be reapplied for each new set of observed financial data and each derivative one wants to valuate, it allows to use a \textit{sole} \textit{pre-trained} neural network for \textit{real time} pricing of every financial derivative from a pre-specified class of payoff functions. 
Let us consider some family of financial derivatives defined through payoff functions
\[
\Phi_\theta:\R_+^{nd} \rightarrow \R,~\theta \in \Theta,
\]
which determines the payoff an investor receives at time $t_n$ in case he bought the derivative $\Phi_\theta$ at initial time $t_0$. The payoff depends on the values of $d\in \N$ underlying securities at $n \in \N$ future times $t_1<t_2<\cdots<t_n$, i.e., the derivative depends on each security $\B{S^k}=(S_{t_1}^k,\dots,S_{t_n}^k)$ for $k=1,\dots,d$. The goal is then to determine all possible today's prices for each $\Phi_\theta$ such that a potential investor cannot profit from one of these prices to exploit \emph{model-independent arbitrage}. This notion refers to strategies that involve trading in underlying securities and/or in liquid options which are cost-free and lead to a profit independent of any model assumptions, i.e., for any possible future evolvement of the underlying security, see also \cite{acciaio2016model}.
As a canonical example we consider the class of payoffs associated to basket options,  which are financial derivatives that allow (but not oblige) at a future time to buy a weighted sum of financial assets (with weights denoted by $(w_i^k)_{i,k}$)  at a predetermined strike $L$. Such an option is only executed if it is favorable for the option-holder to do so, which is the case if the difference between the weighted sum and the strike is positive and therefore  the set of payoffs is given by \black 
\begingroup\makeatletter\def\f@size{8}\check@mathfonts
\def\maketag@@@#1{\hbox{\m@th\normalsize\normalfont#1}}
\begin{equation}\label{eq_intro_basket}
\begin{aligned}
\left\{\Phi_\theta(\B{S}^1,\dots,\B{S}^d); \theta \in \Theta \right\}:= \bigg\{&\max\big\{\sum_{i=1}^n\sum_{k=1}^d w_i^k S_{t_i}^k-L,0\big\} \text{ where } \theta:=\left((w_i^k)_{i,k},L\right) \in  \R^{nd}\times \R\bigg\}.
\end{aligned}
\end{equation}
\endgroup
To find the arbitrage-free upper price bounds of a financial derivative $\Phi_\theta$, we consider model-independent super-replication strategies  (also called super-hedging strategies) \black of $\Phi_\theta$, i.e.,  trading \black strategies that lead for every possible evolvement of the underlying securities to a greater or equal outcome than the payoff of $\Phi_\theta$,  which is referred to as the trading strategy super-replicating \black  $\Phi_\theta$. Prices of such strategies need to be at least as high as the price of $\Phi_\theta$, since otherwise the market would admit model-independent arbitrage, which can indeed be seen by buying the strategy and by selling the derivative $\Phi_\theta$ at initial time. Thus, the smallest price among all model-independent super-replication strategies leads to the arbitrage-free upper price bound of $\Phi_\theta$. Analogue, the greatest price among sub-replication strategies yields the arbitrage-free lower price bound.
Moreover, it is a consequence of (adaptions of) the fundamental theorem of asset pricing (see \cite{acciaio2016model} and \cite{delbaen1994general}) that there exist a dual method to approach the valuation problem: One may also consider all martingale models\footnote{ One often refers to martingale models  as \emph{risk-neutral models}, in which an investor is indifferent of either investing in the underlying security or keeping her money in a bank account with constant interest rate (where typically one assumes the interest rate to be zero for simplicity).\black} which are consistent with bid and ask prices of liquidly traded option prices written on the underlying securities $\boldsymbol{S}=(\boldsymbol{S^1},\dots,\boldsymbol{S^d})$ and expire at the future maturities $(t_i)_{i=1,\dots,n}$ as candidate models.
Then, minimizing and maximizing the expectations $\E_\Q[\Phi_\theta(\boldsymbol{S})]$ among all associated  martingale / \black risk-neutral measures $\Q$ of potential models leads to the desired price bounds, compare e.g.~\cite{acciaio2016model} and \cite{cheridito2017duality} for such results in the discrete time model-independent setting.

Given a payoff function $\Phi_\theta$ from a (parametric) set of payoff functions $\left\{\Phi_\theta, \theta \in \Theta\right\}$, for example from the set of basket options as in \eqref{eq_intro_basket}, we use the sub/super-replication method to compute the lower and upper price bound of $\Phi_{\theta}$ for various different sets of financial data and for several choices of $\theta\in \Theta$, i.e., we compute the bounds in dependence of different observed financial data.
The observable market parameters comprised in the financial data include prices of the underlying securities as well as bid and ask prices of liquidly traded call and put options and its associated strikes. After having computed the price bounds for various different sets of financial data, we let, in accordance with the universal approximation theorem from \cite{hornik1991approximation}, a specially designed neural network learn the relationship between observed financial data and the corresponding model-independent price bounds for a parametric family of payoff functions, compare also Figure~\ref{fig_diagram}.

\begin{figure}[h!]
\begin{center}

\subfloat[\label{fig_diagram}]{
\resizebox{.45\textwidth}{!}{
\begin{tikzpicture}
\node[punkt, inner sep=11pt]
 (y) {$\{\boldsymbol{Y_i}\}_{i}=\{\text{Price bounds of }\Phi_\theta \}$};
 \node[above=1cm of y](dummleft){};
 \node[punkt, inner sep=11pt,above=1cm of dummleft](x){$\{\boldsymbol{X_i}\}_i=$ $\{ \text{Market parameters}$ $\text{and }$ $\theta \}~$}
 	edge[pil](y.north);
 \node[beschr,left=0.2cm of dummleft](textleftarrow){Precise algorithm (possibly time consuming)};
 \node[right=2cm of x](dummy1){};
 \node[below=2.2cm of dummy1](dummy1below){};
 \node[right=3.5cm of dummy1below](dummy2){};
 \node[beschr, inner sep=7pt,below=1cm of dummy1](nnlearn){NN learns relation between $\{\boldsymbol{X_i}\}_i$ and $\{\boldsymbol{Y_i}\}_i$.};
 \node[punkt, inner sep=7pt,right=1.5cm of nnlearn](nn){Trained NN};
 \node[right=0cm of x](dummy_arrow1){}
  	edge[pil,bend left=45](nnlearn.north);
 \node[right=0cm of y](dummy_arrow2){}
   	edge[pil,bend right=45](nnlearn.south);
 \node[right = 0cm of nnlearn](dumm_arrow3){} 
  	edge[pil](nn.west);
\end{tikzpicture}
}}
\hspace{0.1cm}
\subfloat[\label{fig_diagram_mot}]{\resizebox{.45\textwidth}{!}{
\begin{tikzpicture}[node distance=1cm, auto,]
\node[punkt, inner sep=7pt]
 (y) {$\{\boldsymbol{Y_i}\}_{i}=\{\text{Price bounds of }\Phi \}$};
 \node[above=1cm of y](dummleft){};
 \node[punkt, inner sep=7pt,above=1cm of dummleft](x){$\{\boldsymbol{X_i}\}_i=\hspace{2cm}\{$Discretized Marginals $(\mu_1,\mu_2) \}$}
 	edge[pil](y.north);
 \node[beschr,left=0.2cm of dummleft](textleftarrow){Linear Programming};
 \node[right=2cm of x](dummy1){};
 \node[below=2.2cm of dummy1](dummy1below){};
 \node[right=3.5cm of dummy1below](dummy2){};
 \node[beschr, inner sep=7pt,below=1cm of dummy1](nnlearn){NN learns relation between $\{\boldsymbol{X_i}\}_i$ and $\{\boldsymbol{Y_i}\}_i$.};
 \node[punkt, inner sep=7pt,right=1.5cm of nnlearn](nn){Trained NN};
 \node[right=0cm of x](dummy_arrow1){}
  	edge[pil,bend left=45](nnlearn.north);
 \node[right=0cm of y](dummy_arrow2){}
   	edge[pil,bend right=45](nnlearn.south);
 \node[right = 0cm of nnlearn](dumm_arrow3){} 
  	edge[pil](nn.west);
\end{tikzpicture}}
}
\caption{(a): Illustration of the presented approach, that is described in detail in Algorithm~\ref{algo_training_nn}, in order to train a neural network (NN) to learn the model-independent price bounds of a derivative $\Phi_\theta$ from a family $\{\Phi_\theta, \theta \in \Theta\}$ in dependence of given market prices. \\
(b): Illustration of Algorithm~\ref{algo_training_mot} which is applied to learn price bounds of MOT problems from marginals. The price bounds contained in $\boldsymbol{Y_i}$ correspond to the solutions of MOT problems, i.e., to $
\inf_{\Q \in \mathcal{M}(\mu_1,\mu_2)}\E_\Q[\Phi(S_{t_1},S_{t_2})] $ and $\sup_{\Q \in \mathcal{M}(\mu_1,\mu_2)}\E_\Q[\Phi(S_{t_1},S_{t_2})]
$, where $\mathcal{M}(\mu_1,\mu_2)$ denotes the set of martingale measures with fixed marginal distributions $\mu_1$ and $\mu_2$, compare also equation \eqref{eq_definition_mot_set}.}
\end{center}
\end{figure}
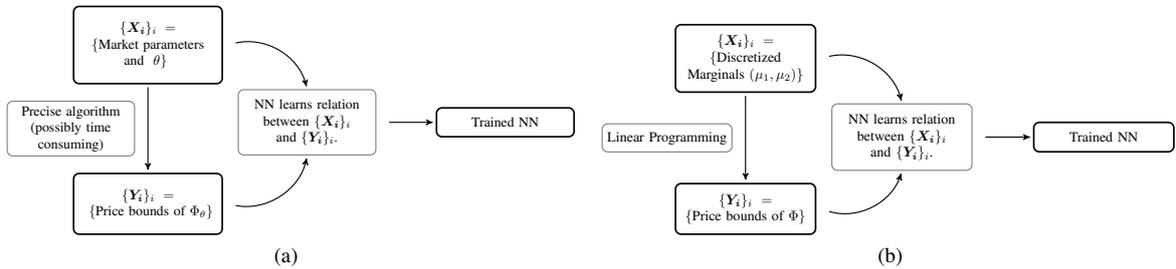
 While there exist \black several numerical routines to compute model-free price bounds of financial derivatives,   
 the only numerical routine that allows to compute model-free price bounds in a purely data-driven approach without imposing any probabilistic assumptions on the market is the approach from \cite{neufeld2020model} which we therefore use to construct a training set of price bounds. Indeed, while \cite{eckstein2019robust}, \cite{eckstein2019computation}, \cite{guo2019computational}, \cite{henry2013automated}, \cite{henry2019martingale}  all provide methods to compute price bounds of financial derivatives, they all rely fundamentally on the assumption that the marginal distributions of each single asset are known exactly. Moreover,
\black 
each established methodology so far requires for every new set of financial data to employ a potentially time-consuming valuation method to find price bounds for every financial derivative of interest.
Our approach circumvents this problem as it enables to train offline a \textit{single} neural network for a \textit{whole family} of related payoff functions, such as e.g.\ basket options with different weights and strikes, and then to determine model-free price bounds in \textit{real time} by using the already trained neural network. Thus, in practice, it suffices to train a couple of neural networks (one for each relevant family of payoffs) and then to use the pre-trained neural networks for valuation-purposes. We refer to  Remark~\ref{rem_approach_in_practice}~(e) \black for further possible examples of parametric families $\{\Phi_\theta, \theta \in \Theta\}$, where we highlight that only a few neural networks are necessary to cover the most relevant payoff functions of financial derivatives.  \black

In Section~\ref{sec_super-hedging}, we first present our approach in a very general setting including multiple assets, multiple time steps, as well as market frictions. To justify our methodology, we show, by proving a continuous relationship between market data and resultant model-free price bounds, that the universal approximation theorem from \cite{hornik1991approximation} is indeed applicable, see Theorem~\ref{thm_convergence}.  { More precisely, under some continuity assumptions on the parametric family of payoff functions $\theta \mapsto \Phi_\theta$ which are typically satisfied for the relevant families of payoff functions in finance (see Section~\ref{sec_market_data}), we prove that both upper and lower arbitrage-free price bound depend continuously on the relevant inputs: strike prices and bid-ask prices of the liquid call and put options, the current prices of the underlying stocks, as well as the parameter determining the payoff function. This, together with the universal approximation property of neural networks proves that a single neural network can indeed learn the arbitrage-free price bounds of a parametric family of payoff functions.} 
To the best of our knowledge, Theorem~\ref{thm_convergence} (as well as Theorem~\ref{lem_approx_mot}) is the first result which proves a continuous relationship between the respective inputs and outputs. This result justifies to learn model-free price bounds by neural networks, and hence provides an important novel contribution to the field.
\black In particular, this means that it is possible to train a single neural network offline on past market data and then to use it online with current market data to compute price bounds of each financial derivative $\Phi_\theta, \theta \in \Theta$. Additionally, we show accuracy and tractability of our presented approach in various high-dimensional relevant examples involving real market data.

In Section~\ref{sec_mot}, we show that the methodology can also be applied to compute two-marginal martingale optimal transport (MOT) problems, see also Figure~\ref{fig_diagram_mot}
for an illustration of the approach where instead of market data entire marginal distributions are the input to the neural network, we refer to Theorem~\ref{lem_approx_mot} for the novel theoretical justification of that approach. 
The knowledge about marginal distributions can be motivated by the findings from \cite{breeden1978prices} which lead to the insight that
complete information about the marginal distributions is equivalent to the knowledge of prices of call options written on the underlying securities for a whole continuum of strikes.

 We further show within several examples the applicability of the presented approach. Mathematical proofs of the theoretical results are provided in Section~\ref{appendix_proof}.

\section{Approximating model-free price bounds with neural networks}\label{sec_super-hedging}
In this section we present an arbitrage-free approach to determine model-free price bounds of a possibly high-dimensional financial derivative when real market data is given. In addition to prices of underlying securities we observe bid and ask prices of call and put options written on these securities , where bid and ask prices refer to the quotations for which the options can be sold and bought\black. Moreover, we explain how model-independent price bounds can be approximated through neural networks.

\subsection{Model-independent valuation of derivatives}
We consider at the present time $t_0 \in [0,\infty)$ $d\in \N$ underlying securities and $n \in \N$ future times $t_0<t_1<\cdots<t_n<\infty$, i.e., the underlying process is given by
\[
\boldsymbol{S}:=\left(\B{S^1},\dots,\B{S^d}\right)=\left(\B{S_{t_1}},\dots, \B{S_{t_n}}\right)=(S_{t_i}^k)_{i=1,\dots,n}^{k=1,\dots,d}
\]
with $\B{S^k} := (S^k_{t_i})_{i=1,\dots,n}$ denoting the $k$-th underlying security and $\B{S_{t_i}}:= (S^k_{t_i})_{k=1,\dots,d}$ denoting the values of the underlying securities at time $t_i$. The process $\boldsymbol{S}$ is modelled as the canonical process on $\R_+^{nd}$ equipped with the Borel $\sigma$-algebra denoted by $\mathcal{B}(\R_+^{nd})$, i.e., for all $i=1,\dots,n$, $k=1,\dots,d$ we have 
\begingroup\makeatletter\def\f@size{9}\check@mathfonts
\def\maketag@@@#1{\hbox{\m@th\normalsize\normalfont#1}}
\[
S_{t_i}^k(\B{s})=s_{i}^k \text{ for all } \B{s}=(s_1^1,\dots,s_1^d,\dots,s_n^1,\dots,s_n^d)\in \R_+^{nd}.
\]
\endgroup
As we want to consider real market data, we cannot - as usual in a vast majority of the mathematical literature on model-independent finance - neglect bid-ask spreads as well as transaction costs. 
Thus, we assume that option prices do not necessarily coincide for buyer and seller, instead we take into account a bid price and an ask price. Let $k \in \{1,\dots,d\},~i \in \{1,\dots,n\}$, $j \in \{1,\dots,n_{ik}^{\operatorname{opt}}\}$, where $n_{ik}^{\operatorname{opt}}$ denotes the amount of tradable put and call options\footnote{We assume the same amount of traded put and call options. This simplifies the presentation, but can without difficulties be extended to a more general setting.} with maturity $t_i$ written on $S^k_{t_i}$. Then a call option on $\B{S^k}$ with maturity $t_i$ for strike $K_{ijk}^{\operatorname{call}} \in  \R_+$ can be bought at price ${\pi^+_{\operatorname{call},i,j,k}}$ and be sold at price ${\pi^-_{\operatorname{call},i,j,k}}$.  As a call option entitles the owner of the option to \emph{buy} the underlying security at price $K_{ijk}^{\operatorname{call}}$ at time $t_i$ it is only exercised if the difference between underlying security and strike $K_{ijk}^{\operatorname{call}}$ is positive, and therefore possesses the payoff $\max\left\{S_{t_i}^k-K_{ijk}^{\operatorname{call}},0\right\}$.   \black Similarly, bid and ask prices for traded put options are denoted by ${\pi^-_{\operatorname{put},i,j,k}}$ and ${\pi^+_{\operatorname{put},i,j,k}}$, respectively.    Put options give the right to \emph{sell} the underlying security at price $K_{ijk}^{\operatorname{put}}$ at time $t_i$. Hence, put options are only exercised if the difference between strike $K_{ijk}^{\operatorname{call}}$  and underlying security is positive, leading to the payoff $\max\left\{K_{ijk}^{\operatorname{put}}-S_{t_i}^k,0\right\}$.  \black \\
Moreover, we assume proportional transaction costs, similar to the approaches in \cite[Section 3.1.1.]{cheridito2017duality} and \cite{dolinsky2014robust}. This means, at each time $t_i$, after having observed the values $\B{{S}_{t_1}},\dots,\B{{S}_{t_i}}$,  rearranging a dynamic self-financing\footnote{  Self-financing means that at any time there is neither consumption nor any money injection. The profit of the trading strategy is purely a consequence of the trading in the underlying security. \black } trading position in the underlying security from\footnote{For $m,n \in \N$ and some set $K \subseteq \R^m$, we denote  by $B\left(K,\R^n\right)$ the set of all functions $f:K \rightarrow \R^n$ which are $\mathcal{B}(K)$/$\mathcal{B}(\R^n)$-measurable, whereas $C(K,\R^n)$ denotes the set of all continuous functions $f:K\rightarrow \R^n$.} $\Delta_{i-1 }^k \in B\left(\R_+^{(i-1)d},\R\right)$, which was the trading position after having observed only  the values $\B{{S}_{t_{1}}},\dots,\B{{S}_{t_{i-1}}}$, to a new trading position $\Delta_{i}^k \in B(\R_+^{id},\R)$, causes transaction costs\footnote{Here also different approaches to measure transaction costs would have been possible. Compare for example the presentations in \cite{buehler2019deep} and \cite{cheridito2017duality}.} of 
$$
\kappa |S_{t_i}^k| \left|\Delta_{i}^k(\boldsymbol{S}_{t_i},\dots,\boldsymbol{S}_{t_1})-\Delta_{i-1}^k(\boldsymbol{S}_{t_{i-1}},\dots,\boldsymbol{S}_{t_1})\right|
$$ 
for some fixed $\kappa \geq  0$. We denote for each $k=1,\dots,d$ by $S_{t_0}^k \in \R_+$ the observable and therefore deterministic current value of the $k$-th security, also called the \emph{spot price} of $\B{S^k}$. Then, given spot prices $\B{S_{t_0}}=(S_{t_0}^1,\dots,S_{t_0}^d)$ and strikes $\boldsymbol{K}:=\left((K_{ijk}^{\operatorname{call}})_{i,j,k}, (K_{ijk}^{\operatorname{put}})_{i,j, k}\right)$, we consider  trading \black strategies with profits of the form\footnote{To simplify the presentation we assume zero interest rates and zero dividend yields.}
\begingroup\makeatletter\def\f@size{8}\check@mathfonts
\def\maketag@@@#1{\hbox{\m@th\normalsize\normalfont#1}}
\begin{equation}\label{eq_phi}
\begin{aligned}
\Psi^{(\boldsymbol{K},\B{S_{t_0}})}_{(a,\B{c_{ijk}},\B{p_{ijk}},\Delta_i^k)}(\boldsymbol{S}):=
a&+\sum_{i=1}^n\sum_{k=1}^d\sum_{j=1}^{n_{ik}^{\operatorname{opt}}}\left({c_{ijk}^+}-{c_{ijk}^-}\right)\max\left\{S_{t_i}^k-K_{ijk}^{\operatorname{call}},0\right\}+\sum_{i=1}^n\sum_{k=1}^d\sum_{j=1}^{n_{ik}^{\operatorname{opt}}}\left({p_{ijk}^+}-{p_{ijk}^-}\right)\max\left\{K_{ijk}^{\operatorname{put}}-S_{t_i}^k,0\right\}\\
&+\sum_{k=1}^d\sum_{i=0}^{n-1} \bigg(\Delta_i^k(\boldsymbol{S}_{t_i},\dots,\boldsymbol{S}_{t_1})\left(S_{t_{i+1}}^k-S_{t_i}^k\right)-\kappa|S_{t_i}^k|\left|\Delta_{i}^k(\boldsymbol{S}_{t_i},\dots,\boldsymbol{S}_{t_1})-\Delta_{i-1}^k(\boldsymbol{S}_{t_{i-1}},\dots,\boldsymbol{S}_{t_1})\right|\bigg)
\end{aligned}
\end{equation}
\endgroup
for an amount of cash $a \in \R$, non-negative long positions $c_{ijk}^+,p_{ijk}^+ \in \R_+$ and non-negative short positions $c_{ijk}^-,p_{ijk}^-\in \R_+$ in call and put options, respectively, for $j=1,\dots,n_{ik}^{\operatorname{opt}}$, $i=1,\dots,n$, $k=1,\dots,d$. In equation \eqref{eq_phi} and for the rest of the paper we use the abbreviations $\B{c_{ijk}}=(c_{ijk}^+,c_{ijk}^-) \in \R_+^2$ and $\B{p_{ijk}}=(p_{ijk}^+,p_{ijk}^-) \in \R_+^2$. Further, the strategies involve self-financing trading positions $\Delta_i^k \in B(\R_+^{id},\R)$ with the convention $\Delta_0^k \in \R$, i.e., to be deterministic, as well as $\Delta_{-1}^k :\equiv 0$. The costs for setting up the position $\Psi^{(\boldsymbol{K},\B{S_{t_0}})}_{(a,\B{c_{ijk}},\B{p_{ijk}},\Delta_i^k)}$ with respect to the bid-ask prices 
\begin{align*}
\boldsymbol{\pi}:=\bigg(&\left({\pi}_{\operatorname{call},i,j,k}^-\right)_{i,j,k},\left({\pi}_{\operatorname{call},i,j,k}^+\right)_{i,j,k},\left({\pi}_{\operatorname{put},i,j,k}^-\right)_{i,j,k},\left({\pi}_{\operatorname{put},i,j,k}^+\right)_{i, j, k}\bigg)
\end{align*}
are given by
\begin{equation}\label{eq_definition_cost_functional}
\begin{aligned}
\mathcal{C}\bigg(\Psi^{(\boldsymbol{K},\B{S_{t_0}})}_{(a,\B{c_{ijk}},\B{p_{ijk}},\Delta_i^k)},\boldsymbol{\pi}\bigg):=
a&+\sum_{i=1}^{n}\sum_{k=1}^d\sum_{j=1}^{n_{ik}^{\operatorname{opt}}} \left(c_{ijk}^+\pi_{\operatorname{call},i,j,k}^+-c_{ijk}^-\pi_{\operatorname{call},i,j,k}^-\right) +\sum_{i=1}^{n}\sum_{k=1}^d\sum_{j=1}^{n_{ik}^{\operatorname{opt}}} \left(p_{ijk}^+\pi_{\operatorname{put},i,j,k}^+-p_{ijk}^-\pi_{\operatorname{put},i,j,k}^-\right).
\end{aligned}
\end{equation}
For a strategy $\Psi^{(\boldsymbol{K},\B{S_{t_0}})}_{(a,\B{c_{ijk}},\B{p_{ijk}},\Delta_i^k)}$ with parameters $(a,\B{c_{ijk}},\B{p_{ijk}},\Delta_i^k)_{i,j,k}$ we introduce the function
\begingroup\makeatletter\def\f@size{9}\check@mathfonts
\def\maketag@@@#1{\hbox{\m@th\normalsize\normalfont#1}}
\begin{align*}
\Sigma(\B{c_{ijk}},\B{p_{ijk}},\Delta_i^k):=&\sum_{i=1}^{n}\sum_{k=1}^d\sum_{j=1}^{n_{ik}^{\operatorname{opt}}}(c_{ijk}^++c_{ijk}^-+p_{ijk}^++p_{ijk}^-)+\sum_{k=1}^d |\Delta_0^k|+\sum_{i=1}^{n-1}\sum_{k=1}^d\|\Delta_i^k\|_{\infty},
\end{align*}
\endgroup
where $\|\cdot\|_{\infty}$ denotes the supremum norm. Imposing a universal upper bound  on the function $\Sigma$, i.e., $\Sigma(\cdot)\leq \mathfrak{B}<\infty$ for some  $\mathfrak{B} \in \R_+$, relates to a restriction on the maximal position an investor is willing/allowed to invest. We want to valuate a derivative with payoff $\Phi \in B(\R_+^{nd},\R)$.
Hence, given strikes $\boldsymbol{K}$, spot prices $\B{S_{t_0}}$, and bid-ask prices $\boldsymbol{\pi}$, our goal is to solve the following super-hedging problem
%\footnote{The inequality $\Psi^{(\boldsymbol{K},\B{S_{t_0}})}_{(a,\B{c_{ijk}},\B{p_{ijk}},\Delta_i^k)} \geq \Phi$ means $\Psi^{(\boldsymbol{K},\B{S_{t_0}})}_{(a,\B{c_{ijk}},\B{p_{ijk}},\Delta_i^k)}(\B{s}) \geq \Phi(\B{s})$ for all $\B{s} \in \R_+^{nd}$.}
\begingroup\makeatletter\def\f@size{8}\check@mathfonts
\def\maketag@@@#1{\hbox{\m@th\normalsize\normalfont#1}}
\begin{equation}\label{eq_subhedging_bidask}
\begin{aligned}
\overline{D}^{\mathfrak{B},B}_{(\boldsymbol{K},\boldsymbol{\pi},\B{S_{t_0}})}\left(\Phi\right):=\inf_{\substack{a\in \R, \\ \B{c_{ijk}} ,
\B{p_{ijk}}\in \R^2_+,\\ (\Delta_i^k) \in B(\R_+^{id},\R)}}&\bigg\{\mathcal{C}\left(\Psi^{(\boldsymbol{K},\B{S_{t_0}})}_{(a,\B{c_{ijk}},\B{p_{ijk}},\Delta_i^k)},\boldsymbol{\pi}\right)\text{ s.t. }\Psi^{(\boldsymbol{K},\B{S_{t_0}})}_{(a,\B{c_{ijk}},\B{p_{ijk}},\Delta_i^k)}(\B{s}) \geq \Phi(\B{s}) \text{ for all } \B{s} \in [0,B]^{nd},\\
&\hspace{7cm} \text{ and } \Sigma(\B{c_{ijk}},\B{p_{ijk}},\Delta_i^k) \leq \mathfrak{B}\bigg\},
\end{aligned}
\end{equation}
\endgroup
for some bounds $\mathfrak{B},B \in (0,\infty]$, where the bound $B$ corresponds to a restriction of the form $S_{t_i}^k\leq B$ for all $i,k$. It is economically reasonable to assume a large but finite $B$ for the securities under consideration, since it imposes no severe restriction and reduces artificial high prices which were not realistic in practice.\footnote{We still allow a priori $\mathfrak{B},B=\infty$ in case one does not want to make restrictions on the trading strategies or exclude unbounded price paths.} A solution of \eqref{eq_subhedging_bidask} defines the largest model-independent arbitrage-free price of the derivative $\Phi$ and simultaneously comes with a strategy that enables to exploit arbitrage if prices for $\Phi$ lie above this price bound.

In analogy to \eqref{eq_subhedging_bidask}, the smallest model-independent arbitrage-free price of $\Phi$ is given by the corresponding sub-hedging problem
\begingroup\makeatletter\def\f@size{8}\check@mathfonts
\def\maketag@@@#1{\hbox{\m@th\normalsize\normalfont#1}}
\begin{align*}
\underline{D}^{\mathfrak{B},B}_{(\boldsymbol{K},\boldsymbol{\pi},\B{S_{t_0}})}\left(\Phi\right):=\sup_{\substack{a\in \R, \\ \B{c_{ijk}} ,
\B{p_{ijk}}\in \R^2_+,\\ (\Delta_i^k) \in B(\R_+^{id},\R)}}&\bigg\{\mathcal{C}\left(\Psi^{(\boldsymbol{K},\B{S_{t_0}})}_{(a,\B{c_{ijk}},\B{p_{ijk}},\Delta_i^k)},\boldsymbol{\pi}\right)\text{ s.t. }\Psi^{(\boldsymbol{K},\B{S_{t_0}})}_{(a,\B{c_{ijk}},\B{p_{ijk}},\Delta_i^k)}(\B{s}) \leq \Phi(\B{s}) \text{ for all } \B{s} \in [0,B]^{nd},\\
&\hspace{7cm} \text{ and } \Sigma(\B{c_{ijk}},\B{p_{ijk}},\Delta_i^k) \leq \mathfrak{B}\bigg\}. %\\
%&=-\overline{D}^{\mathfrak{B},B}_{(\boldsymbol{K},\boldsymbol{\pi},\B{S_{t_0}})}\left(-\Phi\right).
\end{align*}
\endgroup

%\begin{rem}~\label{rem_setting}
%Under some regularity assumptions on $\Phi$, using classical super-hedging duality results, one can show that 
%\[
%\underline{D}^{\mathfrak{B},B}_{(\boldsymbol{K},\boldsymbol{\pi},\B{S_{t_0}})}\left(\Phi\right)=-\overline{D}^{\mathfrak{B},B}_{(\boldsymbol{K},\boldsymbol{\pi},\B{S_{t_0}})}\left(-\Phi\right),
%\]
%see for example \cite{neufeld2020model}.
%\end{rem}

\subsection{Training a neural network for option valuation}\label{sec_training}
Next, we focus on the supervised learning approach we pursue in this paper.
This approach is implemented using neural networks, thus we start this section with a short exposition on neural networks which can be found in similar form in \cite{aquino2019bounds}, \cite{baes2019low}, \cite{buehler2019deep},~\cite{eckstein2019robust}, \cite{eckstein2019computation}, \cite{eckstein2020robust}, or in every standard textbook on the topic (e.g. \cite{bengio2009learning}, \cite{goodfellow2016deep}, or \cite{hassoun1995fundamentals}).
\subsubsection{Neural networks}
In the following we consider a fully-connected neural network which is for input dimension $d_{\operatorname{in}} \in \N$, output dimension $d_{\operatorname{out}} \in \N$, and number of layers $l \in \N$ defined as a function of the form
\begin{equation}\label{eq_nn_function}
\begin{aligned}
\R^{d_{\operatorname{in}}} &\rightarrow \R^{d_{\operatorname{out}}}\\
\boldsymbol{x} &\mapsto \boldsymbol{A_l} \circ \boldsymbol{\varphi}_l \circ \boldsymbol{A_{l-1}} \circ \cdots \circ \boldsymbol{\varphi}_1 \circ \boldsymbol{A_0}(\boldsymbol{x}),
\end{aligned}
\end{equation}
where $(\boldsymbol{A_i})_{i=0,\dots,l}$ are functions of the form 
%%%%%%%%%%%%%%%%%% WIEDER REINMACHEN!
\begingroup\makeatletter\def\f@size{9}\check@mathfonts
\def\maketag@@@#1{\hbox{\m@th\normalsize\normalfont#1}}
\begin{equation}\label{eq_A_i_def} 
\begin{aligned}
\B{A_0}: \R^{d_{{\operatorname{in}}}} \rightarrow \R^{h_1},~\B{A_i}:\R^{h_i}\rightarrow \R^{h_{i+1}}\text{ for } i =1,\dots,l-1, \text{(if } l>1),~\B{A_l} : \R^{h_l} \rightarrow \R^{d_{\operatorname{out}}},
\end{aligned}
\end{equation}
\endgroup
and where for  $i=1,\dots,l$ we have $\boldsymbol{\varphi}_i(x_1,\dots,x_{h_i})=\left(\varphi(x_1),\dots,\varphi(x_{h_i})\right)$,  with $\varphi:\R \rightarrow \R$ being a non-constant function called \emph{activation function}. Here $\boldsymbol{h}=(h_1,\dots,h_{l}) \in \N^{l}$ denotes the dimensions (the number of neurons) of the hidden layers, also called \emph{hidden dimension}. Moreover, for all $i=0,\dots,l$, the function $\B{A_i}$ is assumed to have an affine structure of the form
\[
\boldsymbol{A_i}(\boldsymbol{x})=\boldsymbol{M_i} \boldsymbol{x} + \boldsymbol{b_i}
\]
for some matrix $\boldsymbol{M_i} \in \R^{ h_{i+1} \times h_{i}}$ and some vector $\boldsymbol{b_i}\in \R^{h_{i+1}}$, where $h_0:=d_{\operatorname{in}}$ and $h_{l+1}:=d_{\operatorname{out}}$. 
 We then denote by $\mathfrak{N}_{d_{\operatorname{in}},{d_{\operatorname{out}}}}^{l,\boldsymbol{h}}$ the set of all neural networks with input dimension ${d_{\operatorname{in}}}$, output dimension ${d_{\operatorname{out}}}$, $l$ hidden layers, and hidden dimension $\boldsymbol{h}$.
 Moreover, we consider the set of all neural networks with input dimension $d_{\operatorname{in}}$, output dimension $d_{\operatorname{out}}$, a fixed amount of $l$ hidden layers, but unspecified hidden dimension
 \[
 \mathfrak{N}_{d_{\operatorname{in}},{d_{\operatorname{out}}}}^{l}:=\bigcup_{\boldsymbol{h} \in \N^l}\mathfrak{N}_{d_{\operatorname{in}},{d_{\operatorname{out}}}}^{l,\boldsymbol{h}},
 \]
as well as the set of all neural networks mapping from $\R^{d_{\operatorname{in}}}$ to $\R^{d_{\operatorname{out}}}$ with an unspecified amount of hidden layers
\[
 \mathfrak{N}_{d_{\operatorname{in}},{d_{\operatorname{out}}}}:=\bigcup_{l \in \N}\mathfrak{N}_{{d_{\operatorname{in}}},{d_{\operatorname{out}}}}^{l}.
\]
One fundamental result that is of major importance for the approximation of functions through neural networks is the universal approximation theorem from e.g. \cite[Theorem 2]{hornik1991approximation}, stating that, given some mild assumption on the activation function $\varphi$, every continuous function can be approximated arbitrarily well by neural networks on compact subsets.
\begin{prop}[Universal approximation theorem for continuous functions~\cite{hornik1991approximation}]\label{lem_universal}
Assume that $\varphi \in C(\R,\R)$ and that $\varphi$ is not constant, then for any compact $\K \subset \R^{d_{\operatorname{in}}} $ the set $\mathfrak{N}_{d_{\operatorname{in}},{d_{\operatorname{out}}}}|_{\K}$ is dense in ${C}(\K,\R^{d_{\operatorname{out}}})$ w.r.t. the topology of uniform convergence on $C(\K,\R^{d_{\operatorname{out}}})$.
\end{prop}

Popular examples for activation functions are the \emph{ReLU} function given by $\varphi(x):=\max\{x,0\}$ or the logistic function $\varphi(x):=1/(1+e^{-x})$, which fulfil the assumptions of Proposition~\ref{lem_universal}. Further, we remark that the original statement from \cite[Theorem 2]{hornik1991approximation} only covers output dimension ${d_{\operatorname{out}}}=1$, and $l=1$ hidden layer, but can indeed be generalized to the above statement, compare e.g. \cite[Theorem 3.2.]{kidger2020universal}.

\subsubsection{Approximation of the super-replication functional through neural networks}
We consider for $i=1,\dots,\mathcal{S}$, where $\mathcal{S}\in\N$ denotes the number of samples, input data of the form
$$
\B{X_i}=\left(\boldsymbol{K},\boldsymbol{\pi},\B{S_{t_0}},\theta  \right)
$$
and we aim at predicting via an appropriately trained neural network the following target
$$
 \B{Y_i} = \left(\underline{D}^{\mathfrak{B},B}_{(\boldsymbol{K},\boldsymbol{\pi},\B{S_{t_0}})}\left(\Phi_\theta\right),~\overline{D}^{\mathfrak{B},B}_{(\boldsymbol{K},\boldsymbol{\pi},\B{S_{t_0}})}\left(\Phi_\theta\right)\right),
$$
for a parametrized family $\{\Phi_\theta, \theta \in \Theta\}$.
 If we are additionally interested in predicting the optimal super-replication strategy, then $\B{\widetilde{Y_i}}$ contains instead the associated parameters of the strategies, i.e.,
\[
\B{\widetilde{Y_i}} = \left(a,(\B{c_{ijk}})_{i,j,k},(\B{p_{ijk}})_{i,j,k},(\Delta_i^k)_{i,k} \right),
\]
for the minimal super-replication strategy $\Psi^{(\boldsymbol{K},\B{S_{t_0}})}_{(a,\B{c_{ijk}},\B{p_{ijk}},\Delta_i^k)}$ (and analogue also for the maximal sub-replication strategy), which implicitly also contains the minimal super-replication price by calculating the corresponding cost using \eqref{eq_definition_cost_functional}. However, after having trained a neural network to predict $\B{\widetilde{Y_i}}$ given market data $\B{X_i}$, due to a different training error, the implied price bounds are expected to differ to a larger extent from  $\B{Y_i}$ than those from a neural network which directly predicts the prices $\B{Y_i}$, see also Example~\ref{exa_real_example_strategies} in which we compare both approaches.

According to Proposition~\ref{lem_universal}, a trained neural network can be used to predict price bounds and optimal strategies if price bounds and strategies, respectively, are continuous functions of the input, i.e., of $\B{X_i}$.
The following result stated in Theorem~\ref{thm_convergence} ensures that, under mild assumptions which we discuss subsequently in Remark~\ref{rem_assumptions_thm}, this requirement is fulfilled. For this, we denote for all $k \in \N$  by  $\|\cdot \|_k$ some norm on $\R^k$. Since all norms on Euclidean spaces are equivalent, the specific choice of the norm is irrelevant for the following assertions. The induced metric for $k\in \N$ is denoted by $d_k(\boldsymbol{x},\boldsymbol{y})=\|\boldsymbol{x}-\boldsymbol{y}\|_k$. Moreover, we define for every $B \in (0,\infty)$ the norm $\|f\|_{\infty,B}:=\sup_{x \in [0,B]^{nd}}|f(x)|$ and $d_{\infty,B}(f,g):=\|f-g\|_{\infty,B}$ for $f,g \in C(\R_+^{nd},\R)$. In the case $B=\infty$ we set
\[
d_{\infty,\infty}(f,g):=\frac{\|f-g\|_{\infty}}{1+\|f-g\|_{\infty}}.
\]
 
To the best of our knowledge, the following Theorem~\ref{thm_convergence} proves for the first time a continuous relation between the market inputs and the corresponding price bounds. This novel result justifies to apply neural networks to determine model-independent price bounds and hence provides a significant contribution to the literature. \black 
\begin{thm}\label{thm_convergence}
Let $B\in (0,\infty],\mathfrak{B} \in (0,\infty)$, and $M:=\sum_{i=1}^{n}\sum_{k=1}^d n_{ik}^{\operatorname{opt}}$. 
 Let $\{\Phi_\theta, \theta \in \Theta\},$ for some $\Theta \subset \R^p$ and $p \in \N$, be a (parametric) family of functions in $C(\R_+^{nd},\R)$ such that 
\begin{equation}\label{eq_defn_theta_map}
\begin{aligned}
\left(\Theta,d_p\right) &\rightarrow \left(C(\R_+^{nd},\R),d_{\infty,B}\right)\\
\theta &\mapsto \Phi_\theta
\end{aligned}
\end{equation}
is continuous and let $N_{\operatorname{input}}:=2M+4M+d+p$. Then, the following holds.

\begin{itemize}
\item[(a)]
Let $\mathbb{K}_1 \subset \R_+^{2M}\times \R^{4M}\times \R_+^d\times \Theta$ be a compact set such that both
\begin{equation}\label{eq_no_arbitrage}
\underline{D}^{\mathfrak{B},B}_{(\boldsymbol{K},\boldsymbol{\pi},\B{S_{t_0}})}\left(\Phi_{\theta}\right),\overline{D}^{\mathfrak{B},B}_{(\boldsymbol{K},\boldsymbol{\pi},\B{S_{t_0}})}\left(\Phi_{\theta}\right)\in (-\infty,\infty)
\end{equation}
holds for all $(\boldsymbol{K}, \B{\pi}, \B{S_{t_0}}, \theta) \in \mathbb{K}_1$. Then, the map 
\begingroup\makeatletter\def\f@size{8}\check@mathfonts
\def\maketag@@@#1{\hbox{\m@th\normalsize\normalfont#1}}
\begin{align*}
\left(\mathbb{K}_1,d_{N_{\operatorname{input}}}\right) &\rightarrow \left(\R^2,d_2\right)\\
\left(\boldsymbol{K}, \B{\pi}, \B{S_{t_0}}, \theta\right) &\mapsto \left(\underline{D}^{\mathfrak{B},B}_{(\boldsymbol{K},\boldsymbol{\pi},\B{S_{t_0}})}\left(\Phi_{\theta}\right),\overline{D}^{\mathfrak{B},B}_{(\boldsymbol{K},\boldsymbol{\pi},\B{S_{t_0}})}\left(\Phi_{\theta}\right)\right)
\end{align*}
\endgroup
is continuous.
\item[(b)]
Let $\mathbb{K}_1$ be defined as in (a). Then, for all $\varepsilon>0$ there exists a neural network $\mathcal{N}_1 \in \mathfrak{N}_{N_{\operatorname{input}},2}$ such that for all $(\boldsymbol{K}, \B{\pi},\B{S_{t_0}}, \theta) \in \mathbb{K}_1$ it holds
\begingroup\makeatletter\def\f@size{10}\check@mathfonts
\def\maketag@@@#1{\hbox{\m@th\normalsize\normalfont#1}}
\begin{equation}\label{eq_neural_net_approx}
\begin{aligned}
&\bigg\|\mathcal{N}_1(\boldsymbol{K},\boldsymbol{\pi},\B{S_{t_0}},\theta)-\left(\underline{D}^{\mathfrak{B},B}_{(\boldsymbol{K},\boldsymbol{\pi},\B{S_{t_0}})}\left(\Phi_\theta\right),\overline{D}^{\mathfrak{B},B}_{(\boldsymbol{K},\boldsymbol{\pi},\B{S_{t_0}})}\left(\Phi_\theta\right)\right)\bigg\|_2<\varepsilon.
\end{aligned}
\end{equation}
\endgroup
\item[(c)]
Let $n=1$. Let $\mathbb{K}_2 \subset \R_+^{2M}\times \R^{4M}\times \R_+^d\times \Theta$ be a compact set such that for all $(\boldsymbol{K}, \B{\pi},\B{S_{t_0}},\theta) \in \mathbb{K}_2$ we have that \eqref{eq_no_arbitrage} holds and $\overline{D}^{\mathfrak{B},B}_{(\boldsymbol{K},\boldsymbol{\pi},\B{S_{t_0}})}(\Phi_\theta)$ is  attained by a unique strategy
\[
\left(a^*,(c^*_{1jk})_{j,k},(p^*_{1jk})_{j,k},({\Delta_0^k}^*)_k\right)\left(\boldsymbol{K},\B{\pi}, \B{S_{t_0}}, \theta\right)
\]
satisfying
\begin{equation}\label{eq_boundedness_assumption}
\sup_{(\boldsymbol{K},\B{\pi}, \B{S_{t_0}}, \theta) \in \mathbb{K}_2}\left| a^*\left(\boldsymbol{K}, \B{\pi}, \B{S_{t_0}}, \theta\right)\right|< \infty.
\end{equation}
Then the map
\begingroup\makeatletter\def\f@size{9}\check@mathfonts
\def\maketag@@@#1{\hbox{\m@th\normalsize\normalfont#1}}
\begin{align*}
\left(\mathbb{K}_2,d_{N_{\operatorname{input}}}\right) &\rightarrow \left(\R^{1+4M+d},d_{1+4M+d}\right)\\
\left(\boldsymbol{K},  \B{\pi},\B{S_{t_0}}, \theta\right) &\mapsto \bigg(a^*,(c^*_{1jk})_{j,k},(p^*_{1jk})_{j,k},({\Delta_0^k}^*)_k\bigg)\left(\boldsymbol{K}, \B{\pi}, \B{S_{t_0}}, \theta\right)
\end{align*}
\endgroup
is continuous.
\item[(d)]
Let $n=1$ and let $\mathbb{K}_2$ be defined as in (c). Then, for all $\varepsilon>0$ there exists a neural network $\mathcal{N}_2 \in \mathfrak{N}_{N_{\operatorname{input}},1+4M+d}$ such that for all $(\boldsymbol{K}, \B{\pi}, \B{S_{t_0}}, \theta) \in \mathbb{K}_2$ it holds
\begingroup\makeatletter\def\f@size{7}\check@mathfonts
\def\maketag@@@#1{\hbox{\m@th\normalsize\normalfont#1}}
\begin{equation}\label{eq_neural_net_approx_strats}
\begin{aligned}
&\bigg\|\mathcal{N}_2(\boldsymbol{K},\boldsymbol{\pi},\B{S_{t_0}},\theta)-\left(a^*,(c^*_{1jk})_{j,k},(p^*_{1jk})_{j,k},({\Delta_0^k}^*)_k\right)\left(\boldsymbol{K}, \B{S_{t_0}}, \theta\right)\bigg\|_{1+4M+d}<\varepsilon.
\end{aligned}
\end{equation}
\endgroup
\end{itemize}
\end{thm}

\begin{proof}
See Section \ref{appendix_proof}.
\end{proof} 

\begin{rem}\label{rem_assumptions_thm}
\begin{itemize}
\item[(a)]
Assumption \eqref{eq_no_arbitrage} means that the market with its parameters $\boldsymbol{K}$, $\boldsymbol{\pi}$, $\B{S_{t_0}}$ is arbitrage-free, compare e.g. \cite[Assumption 2.1. and Theorem 2.4.]{neufeld2020model} for the case $n=1$, \cite[Definition 1.1. and Theorem 5.1.]{bouchard2015arbitrage} for the multi-period case with traded options, and  \cite[Theorem 2.1.]{cheridito2017duality} for the general case with market frictions.  Note that assuming an arbitrage-free market is a necessity to determine arbitrage-free price bounds of financial derivatives. Indeed, if the market offers arbitrage, then we can identify a trading strategy fulfilling
$\Psi^{(\boldsymbol{K},\B{S_{t_0}})}_{(a,\B{c_{ijk}},\B{p_{ijk}},\Delta_i^k)}(\B{s}) \geq 0$ for some parameters $(a,\B{c_{ijk}},\B{p_{ijk}},\Delta_i^k)_{i,j,k}$ with price $\mathcal{C}\left(\Psi^{(\boldsymbol{K},\B{S_{t_0}})}_{(a,\B{c_{ijk}},\B{p_{ijk}},\Delta_i^k)},\boldsymbol{\pi}\right)<0$.
Now, consider a super-replication strategy  $(\widetilde{a},\widetilde{\B{c_{ijk}}},\widetilde{\B{p_{ijk}}},\widetilde{\Delta}_i^k)_{i,j,k}$ of some derivative $\Phi$ satisfying $\Psi^{(\boldsymbol{K},\B{S_{t_0}})}_{(\widetilde{a},\widetilde{\B{c_{ijk}}},\widetilde{\B{p_{ijk}}},\widetilde{\Delta}_i^k)}(\B{s}) \geq \Phi(\B{s})$. Then we have  for all $\lambda > 0 $ that 
\[
\Psi^{(\boldsymbol{K},\B{S_{t_0}})}_{(\widetilde{a},\widetilde{\B{c_{ijk}}},\widetilde{\B{p_{ijk}}},\widetilde{\Delta}_i^k)}(\B{s})+\lambda \cdot \Psi^{(\boldsymbol{K},\B{S_{t_0}})}_{(a,\B{c_{ijk}},\B{p_{ijk}},\Delta_i^k)}(\B{s})=\Psi^{(\boldsymbol{K},\B{S_{t_0}})}_{(\widetilde{a}+\lambda a,\widetilde{\B{c_{ijk}}} +\lambda\B{c_{ijk}} ,\widetilde{\B{p_{ijk}}}+\lambda \B{p_{ijk}},\widetilde{\Delta}_i^k+\lambda {\Delta}_i^k)}(\B{s}) \geq \Phi(\B{s})
\]
meaning that $(\widetilde{a}+\lambda a,\widetilde{\B{c_{ijk}}} +\lambda\B{c_{ijk}} ,\widetilde{\B{p_{ijk}}}+\lambda \B{p_{ijk}},\widetilde{\Delta}_i^k+\lambda {\Delta}_i^k)$ is another super-replication strategy, whose price is given by 
\begin{equation}\label{eq_super_rep_lambda}
\mathcal{C}\left(\Psi^{(\boldsymbol{K},\B{S_{t_0}})}_{(\widetilde{a}+\lambda a,\widetilde{\B{c_{ijk}}} +\lambda\B{c_{ijk}} ,\widetilde{\B{p_{ijk}}}+\lambda \B{p_{ijk}},\widetilde{\Delta}_i^k+\lambda {\Delta}_i^k)},\boldsymbol{\pi}\right) = \mathcal{C}\left(\Psi^{(\boldsymbol{K},\B{S_{t_0}})}_{(\widetilde{a},\widetilde{\B{c_{ijk}}},\widetilde{\B{p_{ijk}}},\widetilde{\Delta}_i^k)},\boldsymbol{\pi}\right)+\lambda \cdot \mathcal{C}\left( \Psi^{(\boldsymbol{K},\B{S_{t_0}})}_{(a,\B{c_{ijk}},\B{p_{ijk}},\Delta_i^k)},\boldsymbol{\pi}\right).
\end{equation}
By scaling up $\lambda >0$, we see from \eqref{eq_super_rep_lambda} that the corresponding price decreases, which in turn implies $\overline{D}^{\mathfrak{B},B}_{(\boldsymbol{K},\boldsymbol{\pi},\B{S_{t_0}})}\left(\Phi\right) \ll 0$ (for $\mathfrak{B}$ large enough). 

With an analogue argument we conclude that if the market offers arbitrage, then we have $\underline{D}^{\mathfrak{B},B}_{(\boldsymbol{K},\boldsymbol{\pi},\B{S_{t_0}})}\left(\Phi\right) \gg 0$, preventing the computation of reasonable price bounds.

% \overline{D}^{\mathfrak{B},B}_{(\boldsymbol{K},\boldsymbol{\pi},\B{S_{t_0}})}\left(\Phi\right):=\inf_{\substack{a\in \R, \\ \B{c_{ijk}} ,
%\B{p_{ijk}}\in \R^2_+,\\ (\Delta_i^k) \in B(\R_+^{id},\R)}}&\bigg\{\mathcal{C}\left(\Psi^{(\boldsymbol{K},\B{S_{t_0}})}_{(a,\B{c_{ijk}},\B{p_{ijk}},\Delta_i^k)},\boldsymbol{\pi}\right)\text{ s.t. }\Psi^{(\boldsymbol{K},\B{S_{t_0}})}_{(a,\B{c_{ijk}},\B{p_{ijk}},\Delta_i^k)}(\B{s}) \geq \Phi(\B{s}) \text{ for all } \B{s} \in [0,B]^{nd},\\
%&\hspace{7cm} \text{ and } \Sigma(\B{c_{ijk}},\B{p_{ijk}},\Delta_i^k) \leq \mathfrak{B}\bigg\}
\black 
\item[(b)] In an arbitrage-free market, a necessary requirement for the existence of a unique optimizer, as assumed in Theorem~\ref{thm_convergence}~(c)\black, is that the considered market instruments are non-redundant, i.e., that the payoffs of the market instruments are linear independent. 
In our case, this means that to avoid ambiguity of  minimal super-replication strategies, one should only consider put options that are written on other strikes than the ones for the call options under consideration.
\item[(c)] As a canonical example for a parametric family of payoff functions, we consider for example basket call options with payoffs
\begin{align*}
&\bigg\{\Phi_{\theta}=\max\big\{\sum_{i=1}^n\sum_{k=1}^d w_i^k S_{t_i}^k-L,0\big\}\text{ where }\theta \in \Theta:=\{((w_i^k)_{i,k},L)\} = \R^{nd}\times \R_+\bigg\},
\end{align*}
i.e., the strike $L$ and the weights $(w_i^k)_{i,k}$ are inputs to the trained neural network. 
For any $0<B<\infty$, we have the continuity of the map $$(\Theta,d_{nd+1}) \rightarrow \left(  C(\R^{nd}_+,\R),d_{\infty,B}\right),~\theta \mapsto \Phi_\theta.$$
Thus, we can find a neural network which fulfils \eqref{eq_neural_net_approx} with respect to $\underline{D}^{\mathfrak{B},B}$ and $\overline{D}^{\mathfrak{B},B}$. We remark that assuming a uniform large bound $B$ on the possible values of $S_{t_i}^k$ imposes no severe constraint for practical applications with real market data and allows to reduce the difference between the no-arbitrage price bounds by not considering unbounded prices which are unrealistic in practice. For further examples of parametric families of payoff functions, e.g., best-of-call options or call-on-max options, we refer to \cite[Example 3.2. (i)--(vi)]{neufeld2020model}.

\item[(d)]Theorem~\ref{thm_convergence} is also applicable to a single pre-specified continuous payoff function $\Phi$ when setting $\Phi_\theta = \Phi$ for all $\theta \in \Theta$.
\item[(e)]
Note that we restrict the assertion of Theorem~\ref{thm_convergence}~(d) to $n=1$ to make sure that $\Delta_i^k$, which is the output of the neural network, is a number, not a function.
\item[(f)] An analogue result as in Theorem~\ref{thm_convergence}~(c) and Theorem~\ref{thm_convergence}~(d) for optimal sub-hedging strategies can be obtained in the same way.
{ 
\item[(g)] We highlight that our approach computes \emph{model-independent} price bounds, i.e., no assumptions on underlying financial models are imposed and market prices of call and put options are considered as exogenous inputs. Nevertheless, in a frictionless market, it is possible for each arbitrage-free sample $\left(\boldsymbol{K},  \B{\pi},\B{S_{t_0}}, \theta\right)$ to find a stochastic model (in which call option prices are given endogenously) that is consistent with this data, compare e.g. \cite{davis2007range}. Therefore one could understand the model-independent price bounds obtained with respect to the given sample $\left(\boldsymbol{K},  \B{\pi},\B{S_{t_0}}, \theta\right)$ also as the prices of hedging strategies in such a consistent model. However, note that such a model (expressed by a probability measure) would be different for each considered sample.}
\end{itemize}
\end{rem}

Finally, Algorithm~\ref{algo_training_nn} describes,  relying on the results from Theorem~\ref{thm_convergence}, how one can train a neural network which approximates these price bounds.

\begin{algorithm}[h!]\label{algo_training_nn}
\SetAlgoLined\small{
\SetKwInOut{Input}{Input}
\SetKwInOut{Output}{Output}
\KwData{Call and put option prices (bid and ask) on different securities and maturities; Associated strikes, maturities, and spot prices;}
\Input{Algorithm to compute price bounds of exotic derivatives; Family $\{\Phi_\theta, \theta \in \Theta\}$ of payoff functions $\Phi_{\theta}:\R_+^{nd} \rightarrow \R$ fulfilling the requirements of Theorem~\ref{thm_convergence}; Hyper-parameters of the neural network; Number $n_{\operatorname{subset}}$ of considered functions from $\{\Phi_\theta, \theta \in \Theta\}$ for each sample; Transaction costs $\kappa \geq 0$; Bounds $\mathfrak{B}$ and $B$;}

%\For{$\tilde{t}$ in $\{$Observation Dates$\}$} {We consider call option prices written on $\tilde{d} \geq d$ underlying securities with $\tilde{n} \geq n$ different maturities. \\
\For{each sample $ \left(\boldsymbol{K}_i,\boldsymbol{\pi}_i,{\B{S_{t_0}}_i}\right)$ of data considering exactly $n$ maturities and $d$ securities}{
 $\B{\widetilde{X_i}} \gets \left(\boldsymbol{K}_i,\boldsymbol{\pi}_i,{\B{S_{t_0}}_i}\right)$;}
$\mathcal{S} \gets \#\{\B{\widetilde{X_i}}\}$; \tcp{Assign number of samples and call it $\mathcal{S}$.}
\For{$i$ in $\{1,\dots,\mathcal{S}\}$}
{
Generate a (random) subset ${\{\Phi_{\theta_j}, j=1,\dots,n_{\operatorname{subset}}}\} \subset \{\Phi_\theta, \theta \in \Theta\}$;\\
\For{$j$ in $\{1,\dots,n_{\operatorname{subset}}\}$}{
Compute for $\B{\widetilde{X}}_i$ the corresponding price bounds {\tiny$\left( \underline{D}^{\mathfrak{B},B}_{(\boldsymbol{K}_i,\boldsymbol{\pi}_i,{\B{S_{t_0}}}_i)}\left(\Phi_{\theta_j}\right),  \overline{D}^{\mathfrak{B},B}_{(\boldsymbol{K}_i,\boldsymbol{\pi}_i,{\B{S_{t_0}}}_i)}\left(\Phi_{\theta_j}\right)\right)$};\\
$\boldsymbol{X_{(i-1)n_{\operatorname{subset}}+j}} \gets (\B{\widetilde{X_i}},\theta_j)$;\\
$\boldsymbol{Y_{(i-1)n_{\operatorname{subset}}+j}} \gets \left( \underline{D}^{\mathfrak{B},B}_{(\boldsymbol{K}_i,\boldsymbol{\pi}_i,{\B{S_{t_0}}}_i)}\left(\Phi_{\theta_j}\right),  \overline{D}^{\mathfrak{B},B}_{(\boldsymbol{K}_i,\boldsymbol{\pi}_i,{\B{S_{t_0}}}_i)}\left(\Phi_{\theta_j}\right)\right)$;
\\
\tcp{In this step it would be possible to use any algorithm that can compute these bounds reliably, see Remark~\ref{rem_mot_computation}}
%\item We normalize the data in $X_i$.
%\item We divide the data into training, validation and test set.
}
}
Train with back-propagation (\cite{rumelhart1986learning}) a neural network $\mathcal{N}\in \mathfrak{N}_{d_{\operatorname{in}},2}$ with a sufficient number of neurons and hidden layers such that $\mathcal{N}(\B{X_i}) \approx \B{Y_i}$;
\caption{Training of a neural network via back-propagation for the computation of the price bounds $\underline{D}^{\mathfrak{B},B}_{(\boldsymbol{K},\boldsymbol{\pi},\B{S_{t_0}})}\left(\Phi_\theta\right),\overline{D}^{\mathfrak{B},B}_{(\boldsymbol{K},\boldsymbol{\pi},\B{S_{t_0}})}\left(\Phi_\theta\right)$ of a class of financial derivatives $\{\Phi_\theta\}_{\theta \in \Theta}$.}
\Output{Trained neural network $\mathcal{N} \in \mathfrak{N}_{N_{\operatorname{input}},2}$ with $N_{\operatorname{input}}$ as in Theorem~\ref{thm_convergence};}
}
\end{algorithm}

\begin{rem}\label{rem_mot_computation}
To compute model-independent price bounds given option prices, we can use e.g. a linear programming approach based on grid discretization as proposed in \cite{eckstein2019robust}, \cite{guo2019computational}, or \cite{henry2013automated}. If the payoff function only depends on one future maturity and is continuous piecewise affine (CPWA, see e.g. \cite[Example 3.2.]{neufeld2020model}) we can also use the numerically very efficient algorithm proposed in \cite{neufeld2020model}. If the payoff function fulfils a so-called martingale Spence--Mirrleess condition\footnote{This means that $\frac{\partial^3}{\partial xy^2}\Phi(x,y)$ exists and satisfies $\frac{\partial^3}{\partial xy^2}\Phi(x,y) >0$  for all $x,y$.} one can apply the algorithm presented in \cite{henry2019martingale}. For multiple time-steps, another possibility is to apply the penalization approach presented in \cite{eckstein2019computation}.
 The minimization is then performed using a stochastic gradient descent algorithm with some penalization parameter $\gamma$ which enforces the optimizing strategy to be a super-hedge.
\end{rem}

\subsection{Examples}
\label{sec_market_data}
In this section we present, in selected examples, the results of our approach when applied to real market data.
\subsubsection{Training data}\label{sec_trainingdata}
We consider for the training of all neural networks financial market data received from \emph{Thomson Reuters Eikon} that was observed on $10$th June $2020$. The data includes bid and ask prices on call options written on all $500$ constituents of the American stock market index S$\&$P $500$.
Note that, in Example~\ref{exa_real_example_strategies}, we also predict the optimal super-hedging strategy and not only optimal price bounds. Thus, to avoid ambiguity of the optimal strategy, as explained in Remark~\ref{rem_assumptions_thm}, we do not consider any put options there.
We consider for each constituent and each available maturity of an option the $20$ most liquid strikes, i.e., the bid and ask prices of options with the highest trading volume.
\subsubsection{Test data}\label{sec_testdata}
For testing the trained neural networks we consider - as for the training data - option prices on all constituents of the S$\&$P $500$. The data was observed on $23$rd August $2020$. We highlight that, in particular, the test data comes from a different dataset than the training data.
\subsubsection{Implementation} \label{sec_implementation}
The training of each of the neural networks is performed using the back-propagation algorithm (\cite{rumelhart1986learning}) with an Adam optimizer (\cite{kingma2014adam}) implemented in \emph{Python} using \emph{Tensorflow} (\cite{abadi2016tensorflow}). For the optimization with the Adam optimizer we use a batch size of $256$. The architecture involves a $L^2$-loss function, the neural networks comprise $3$ hidden layers with $512$ neurons each and \emph{ReLU} activation functions. The samples are normalized before training with a min-max scaler. Moreover, we assign $10 \%$ of the training data to a validation set to be able to apply early stopping (compare \cite[Chapter 7.8.]{goodfellow2016deep}) to prevent overfitting to the training data. To reduce the internal covariate shift of the neural network and to additionally regularize it, we apply batch normalization (\cite{ioffe2015batch}) after each layer.
%We state explicitly the parts of the examples where we differ from this architecture.\\
All the codes related to the examples below\footnote{For copyright reasons we can only provide the used code, but we cannot provide the used data.}, as well as the trained neural networks are provided under \href{https://github.com/juliansester/deep\_model\_free\_pricing}{https://github.com/juliansester/deep\_model\_free\_pricing}. For all examples we assume no transaction costs, i.e., we have $\kappa = 0$.

\begin{exa}[Training of the valuation of call options given prices of other call options]\label{exa_call_options}

We want to train the valuation of  call options for arbitrary strikes, i.e., we consider payoff functions from the set
\begin{align*}
\{\Phi_\theta, \theta \in \Theta\}=\bigg\{\Phi_L(S_{t_1}^1):= &\max\left\{S_{t_1}^1-L,0 \right\},\text{ with } L \in \R_+\bigg\}.
\end{align*}
Note that the assumptions of Theorem~\ref{thm_convergence} are met for any $B \in (0,\infty], \mathfrak{B} <\infty$, which we choose therefore large enough to not impose a restriction. Thus Theorem~\ref{thm_convergence}~(b) ensures that we can train a single neural network for the above mentioned parameterized family of payoff functions.

In this example, a single sample $\B{X_i}$ consists of $62$ total entries which comprise $20$ bid prices, $20$ ask prices, $20$ associated strikes of call options, as well as the underlying spot price and the strike $L$ of the call option $\Phi_L$ which we want to price.

We apply Algorithm~\ref{algo_training_nn} to train a neural network, in particular, for each of the prices from the training data, we create several different random strikes $L$, for which we compute a corresponding $\B{Y_i}$ which consists of lower and upper model-independent price bounds for all samples according to the algorithm from \cite{neufeld2020model}. With this methodology we create a training set with $100 000$ samples.
 We then train, as described in Section~\ref{sec_implementation}, a neural network using back-propagation and test it on the test data, described in Section~\ref{sec_testdata}, which was observed at a later date (August $2020$). The test set consists of $10000$ samples.

The results of the training yield a mean absolute error of $2.2033$  as well as a mean squared error of $25.8779$ \black on the test set and are depicted in Figure~\ref{fig_accuracy_call}.
To be able to compare the error independent of the size of the spot price of the underlying security, we report a mean absolute error of $0.0111$ when dividing the predicted prices by the spot prices. We call this value the \emph{relative mean absolute error}.  The corresponding squared distance after division by the spot prices amounts to $0.0003$ and is called \emph{relative mean squared error}. \black 
Compare also  Figure~\ref{fig_accuracy_call_relative}, where we depict the \emph{relative error} of each sample in the test set, i.e., the difference of each prediction from its target value, after division with the spot price.

\begin{figure}[h!]
\begin{center}

\subfloat[\label{fig_accuracy_call}]{
\includegraphics[scale=0.4]{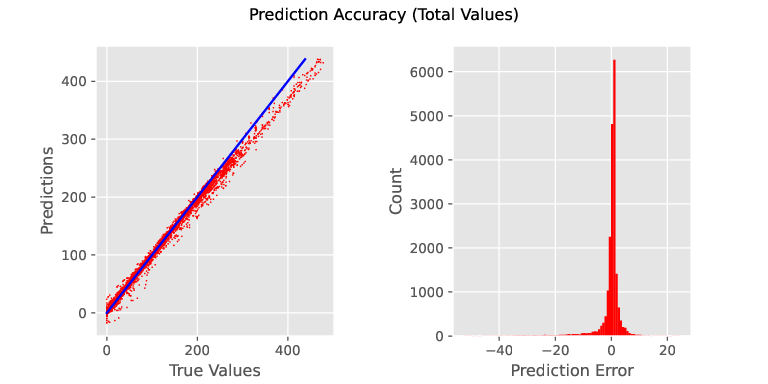}
}
\subfloat[\label{fig_accuracy_call_relative}]{\includegraphics[scale=0.4]{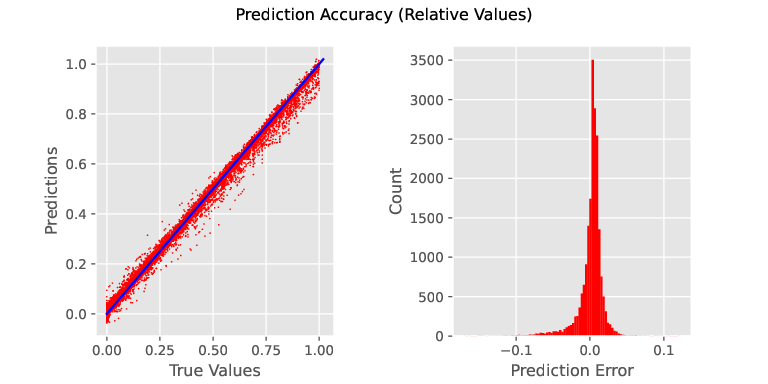}}

\caption{(a): This figure illustrates 
%the progress of learning the model-free price bounds of call options as well as 
the accuracy of the predictions on the test set. The left panel shows a plot of all target values (x-values) and its predictions (y-values), the right panel depicts a histogram of the prediction error, i.e., the error between target values and predicted values.\\
(b):This figure shows the accuracy of the predictions of call option prices on the test set when considering the {relative error}, i.e., when dividing the predicted prices $\B{Y_i}$ by the corresponding spot prices.}
\end{center}
\end{figure}
\end{exa}

\begin{exa}[Optimal strategies of basket options]\label{exa_real_example_strategies}
We consider payoff functions of basket options written on two assets, i.e., the class of payoff functions is defined through
\begin{align*}
\{\Phi_\theta, \theta \in \Theta\}=\bigg\{&\Phi_{w_1,w_2,L}(S_{t_1}^1,S_{t_1}^2)= \max\left\{w_1 S_{t_1}^1+ w_2 S_{t_1}^2-L,0\right\},\text{ with } w_1,w_2,L \in \R_+\bigg\}.
\end{align*}
When $B, \mathfrak{B} <\infty$, the assumptions of Theorem~\ref{thm_convergence}~(b), respectively those of Remark~\ref{rem_assumptions_thm}~(c), are fulfilled. 
For each of the considered market prices from the training set and test set, respectively, we create in accordance with Algorithm~\ref{algo_training_nn} several different weights $w_1,w_2$ and some strike $L$.
Thus, a sample $\B{X_i}$ consists of the spot prices $S_{t_0}^1,S_{t_0}^2$, the generated values $w_1,w_2,L$ as well as of bid and ask prices with associated strikes of both assets, i.e., in total each $\B{X_i}$ consists of $125$ numbers.

We aim at { simultaneously }predicting the minimal super-replication strategy { and} the maximal sub-replication strategy. Therefore, we compute the parameters of the strategies attaining the price bounds according to the algorithm from \cite{neufeld2020model}. Thus, in our case each sample $\B{Y_i}$ comprises $86$ values, which constitute, for both lower and upper bound, of the initial investment $a$ ($1$ parameter), the buy and sell positions in call options\footnote{{ Here, by slight abuse of notation, we denote by $c_{ljk} \in \R$ the net position invested in the option, i.e., $c_{ljk}$ is also allowed to attain negative values.}} $(c_{1jk})_{j,k}{{j=1,\dots,20}, \atop{k=1,2}}$ ($40$ parameters) and the investment positions in the underlying securities $(\Delta_0^k)_{k=1,2}$ ($2$ parameters).

 Moreover, \black training a neural network to learn the relationship between market parameters and optimal super-replication strategy  (without the price bound) \black  is possible due to Theorem~\ref{thm_convergence}~(d), which implies that the difference in absolute values between predicted parameters $\left(a,(\B{c_{1jk}})_{j,k},(\Delta_0^k)_k\right)$ and true parameters $\left(a^*,(\B{c^*_{1jk}})_{j,k},({\Delta_0^k}^*)_k\right)$ of the optimal strategy should not differ significantly after training. 
We train the neural network on $150000$ samples, test it on $10000$ samples, and obtain indeed a small relative mean absolute error of $0.0015$. { We highlight that the training set remains the same, in particular does not consist of trading strategies.}

After having trained the neural network to predict the minimal super-hedging strategies, we are able to derive from these strategies the optimal price  bounds \black using \eqref{eq_definition_cost_functional} and compare it with the predictions from a neural network which is trained on predicting lower and upper price bounds directly instead of predicting optimal strategies. In Figure~\ref{fig_prediction_prices_vs_strategies} we show that  however\black , as expected, the neural network that predicts prices directly performs by far better than the price bound predictions that are derived  via \eqref{eq_definition_cost_functional} from the trained \black  strategies, when evaluated on the test set. The relative\footnote{ Note that here the relative error refers to the error after division with the \emph{weighted} sum of the spot prices, where the weights are determined by the weights in the payoff of the basket option. \black } mean absolute error of the direct prediction of the price bounds is $0.0319$, whereas the relative mean absolute error of the prediction relying on the strategies is $0.1804$.  The corresponding  relative mean squared errors are $0.0148$ and $0.5045$, respectively. \black

\begin{figure}[h!]
\begin{center}
\includegraphics[scale=0.4]{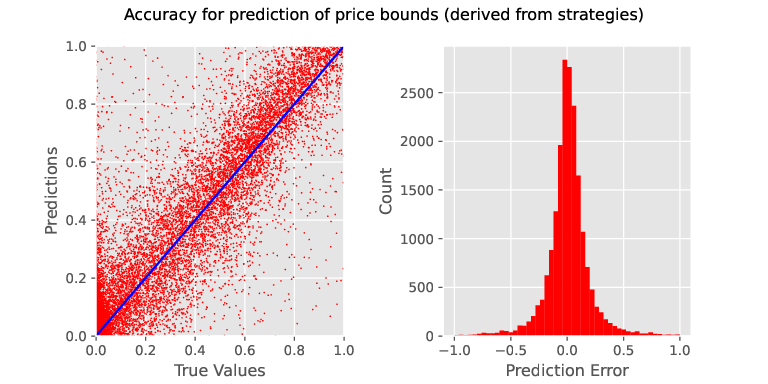}
\includegraphics[scale=0.4]{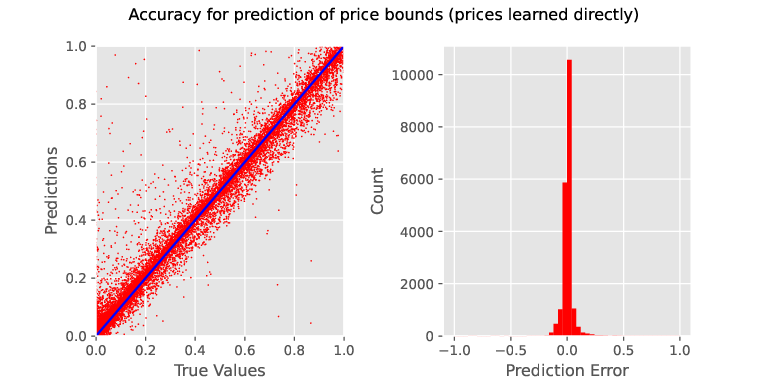}
\caption{This figure compares the accuracy of prediction of price bounds derived from the trained strategies using \eqref{eq_definition_cost_functional} (left) with those predictions from neural networks that are trained to predict the prices directly (right). We depict the relative error of the predictions by dividing the prediction error through the weighted sum of spot prices, where the weights are according to the weights in the payoff of the basket option.}\label{fig_prediction_prices_vs_strategies}
\end{center}
\end{figure}

The larger approximation error when approximating first the strategies by a neural network and then deriving price bounds from this approximation can be explained as follows.

When approximating the price bounds directly, then we have, after sufficient training of a neural network, according to Theorem~\ref{thm_convergence}{~(b)}, a maximal absolute approximation error of order { $\varepsilon_1$} between the upper price bound and the output of the neural network, given a tolerance level of { $\varepsilon_1>0$}.

In contrast, when approximating the optimal super-replication strategy $\bigg(a^*,(c^*_{1jk})_{j,k},(p^*_{1jk})_{j,k},({\Delta_0^k}^*)_k\bigg)\left(\boldsymbol{K}, \B{\pi}, \B{S_{t_0}}, \theta\right)$ by the output of a neural network, denoted by $\bigg({a^{\mathcal{N}\mathcal{N}}},(c^{\mathcal{N}\mathcal{N}}_{1jk})_{j,k},(p^{\mathcal{N}\mathcal{N}}_{1jk})_{j,k},({\Delta_0^k}^{\mathcal{N}\mathcal{N}})_k\bigg)\left(\boldsymbol{K}, \B{\pi}, \B{S_{t_0}}, \theta\right)$ { at some tolerance level $\varepsilon_2>0$, } then the absolute error between the upper price bound $\overline{D}^{\mathfrak{B},B}_{(\boldsymbol{K},\boldsymbol{\pi},\B{S_{t_0}})}\left(\Phi_{\theta}\right) =\mathcal{C}\left(\Psi^{(\boldsymbol{K},\B{S_{t_0}})}_{(a^*,\B{c_{1jk}}^*,\B{p_{1jk}}^*,{\Delta_0^k}^*)},\boldsymbol{\pi}\right)$ and the price $\mathcal{C}\left(\Psi^{(\boldsymbol{K},\B{S_{t_0}})}_{(a^{\mathcal{N}\mathcal{N}},\B{c_{1jk}}^{\mathcal{N}\mathcal{N}},\B{p_{1jk}}^{\mathcal{N}\mathcal{N}},{\Delta_0^k}^{\mathcal{N}\mathcal{N}})},\boldsymbol{\pi}\right)$ derived from the approximated strategy computes by \eqref{eq_definition_cost_functional} as  
\begin{align*}
\Bigg|a^*-a^{\mathcal{N}\mathcal{N}}&+\sum_{k=1}^d\sum_{j=1}^{n_{1k}^{\operatorname{opt}}} \left(({c_{1jk}^+}^*-{c_{1jk}^+}^{\mathcal{N}\mathcal{N}})\pi_{\operatorname{call},1,j,k}^+-({c_{1jk}^-}^*-{c_{1jk}^-}^{\mathcal{N}\mathcal{N}})\pi_{\operatorname{call},1,j,k}^-\right) \\
&+\sum_{k=1}^d\sum_{j=1}^{n_{1k}^{\operatorname{opt}}} \left(({p_{1jk}^+}^*-{p_{1jk}^+}^{\mathcal{N}\mathcal{N}})\pi_{\operatorname{put},1,j,k}^+-({p_{1jk}^-}^*-{p_{1jk}^-}^{\mathcal{N}\mathcal{N}})\pi_{\operatorname{put},1,j,k}^-\right) \Bigg|
\end{align*}
which is, according to \eqref{eq_neural_net_approx_strats}, as large as of order
\[
{ \varepsilon_2} \cdot  \max \left\{ \max_{j,k} \pi_{\operatorname{call},1,j,k}^+,~ \max_{j,k} \pi_{\operatorname{call},1,j,k}^-,~\max_{j,k} \pi_{\operatorname{put},1,j,k}^+,~\max_{j,k} \pi_{\operatorname{put},1,j,k}^-\right\}.
\]
{ However, $C:= \max \left\{ \max_{j,k} \pi_{\operatorname{call},1,j,k}^+,~ \max_{j,k} \pi_{\operatorname{call},1,j,k}^-,~\max_{j,k} \pi_{\operatorname{put},1,j,k}^+,~\max_{j,k} \pi_{\operatorname{put},1,j,k}^-\right\}$ is typically significantly larger than $1$} (the largest call option price in the considered test set was $349 \$ $).
{ Therefore, to obtain similar approximation results for the prices derived from strategies as those when approximating prices directly, one needs to consider a tolerance level of $\varepsilon_2 = \frac{\varepsilon_1}{C}$ which usually is significantly smaller than $\varepsilon_1$. In our case, the training set was not large enough to obtain such a high precision in the approximation of the strategies\footnote{{One certainly could enlarge the training set to improve the approximation results. However, the goal of this example is to showcase that on a \textbf{fixed} training set predicting price bounds directly is more efficient than predicting price bounds via \eqref{eq_definition_cost_functional} after having approximated the optimal strategies.}}.

The above outlined reasoning is supported by a comparatively large approximation error which is already observed on the training set. When predicting prices from strategies via \eqref{eq_definition_cost_functional}, we have a relative mean absolute error of $0.1311$ on the training set, which is only marginally smaller than the relative mean absolute error of $0.1804$ which we reported for predictions on the test set. In comparison, the relative mean absolute error for directly predicting prices is $0.0160$ on the training set, whereas we reported a relative mean absolute error of $0.0319$ on the test set.}

{ We conclude that}, even though it is theoretically possible to derive price bounds with an arbitrarily high precision from approximated strategies, in practice it turns out to be more efficient {(and requires a smaller training set)} to train a neural network that approximates the price bounds directly if one is only interested in the prediction of these.

\end{exa}
 
\begin{exa}[Neural networks trained for basket options applied to call options] \label{exa_basket_to_call}

We reconsider the trained neural network from Example~\ref{exa_real_example_strategies} predicting the price bounds of basket options written on two assets.
\begin{itemize}
\item We now test this neural network on the same test set as in  Example~\ref{exa_call_options} which takes into account call options instead of basket options.  To be able to apply the neural network that takes inputs with $125$ entries, we modify the original samples $\bold{X}_i$ from Example~\ref{exa_call_options} (originally containing $62$ entries) by duplicating the original entries (except for the strike of the call option) and by additionally adding weights of $0.5$ and $0.5$ as well as the original strike, leading to $2 \cdot 61 + 2+ 1 = 125$ entries. This means that we consider a basket option with payoff of the form $\max\{0.5S_{t_1}^1+0.5S_{t_1}^1-K,0\} = \max\{S_{t_1}^1-K,0\} $ which is a call option.
\item  We then observe a relative mean absolute error of $0.0230$ and a relative mean squared error of $0.0023$ (in comparison with $0.0111$ and $0.0003$,  respectively for the neural network from Example~\ref{exa_call_options} that was only trained on call options). 
\item 
This shows that the more general neural network performs reasonably well also on the more specific payoffs, but is less { specifically} trained and therefore is outperformed by the neural network solely trained on call options. 
\item 
To improve the performance of this neural network, we retrain the neural network by adding in addition to the $150000$ samples containing basket options from Example~\ref{exa_real_example_strategies} also the $100000$ samples containing call options from Example~\ref{exa_call_options}.
\item 

Indeed, the result is a trained neural network that performs well on both basket options and call options. On the test set for call options we obtain a relative mean absolute error of $0.0112$ and a relative mean squared error of $0.0004$ (compared to $0.0111$ and $0.0003$ obtained in Example~\ref{exa_call_options}). On the test set for basket options  we compute a relative mean absolute error of $0.0292$ and a relative mean squared error of $0.0374$ (compared to $0.0319$ and $0.0148$ obtained in Example~\ref{exa_real_example_strategies}). See also Figure~\ref{fig_basket_to_call}, where we depict the accuracy of the predictions before and after adding the additional samples. { This shows that training the neural network on an additional task (call option price bounds) did improve notably the approximation quality on this specific task while the approximation quality on the well-known task (basked option price bounds) is barely affected.}
\end{itemize}
{ This example showcases that it is possible to train a single neural network that approximates price bounds reasonably well for different types of payoff functions given that these payoff functions are represented appropriately in the training set. For an overview of different types of payoff functions that can be covered by a single neural network we refer to Table~\ref{tbl_payoff_functions}. Moreover, the approach pursued in Example~\ref{exa_basket_to_call} can be considered as an instance of \emph{transfer learning}, since we used an existing neural network as a starting point to improve the approximation quality on a more specific class of payoff functions.}
 \begin{figure}[h!]
\begin{center}
\includegraphics[scale=0.4]{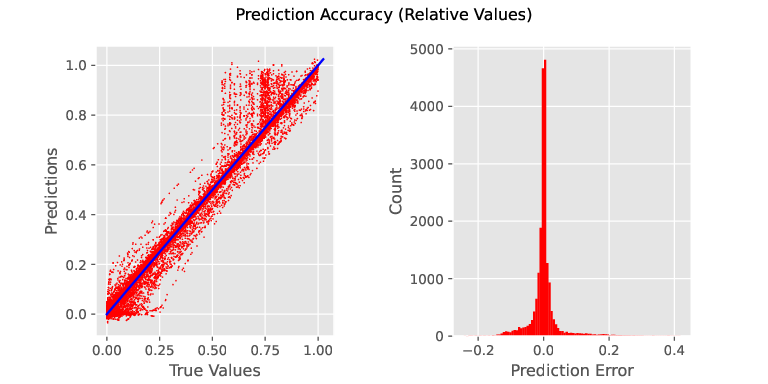}
\includegraphics[scale=0.4]{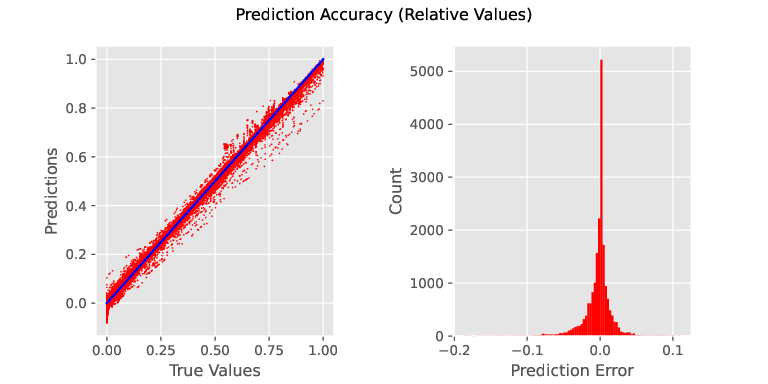}
\caption{ This figure compares the accuracy of prediction of price bounds of call options using the neural network trained solely on basket options (left) with those predictions from a neural network that is trained additionally also on call options (right). We depict the relative error of the predictions by dividing the prediction error { by} the spot prices. \black }\label{fig_basket_to_call}
\end{center}
\end{figure}
\end{exa}
\black

\begin{exa}[High-dimensional payoff function]\label{exa_basket}
We train a neural network to predict prices of a basket option depending on $30$ underlying securities, i.e., we consider a family of payoff functions of the form
\begin{align*}
\{\Phi_\theta, \theta \in \Theta\}=\bigg\{&\Phi_{w_1,\cdots, w_{30},L}(S_{t_1}^1,\cdots, S_{t_1}^{30})= \max\big\{\textstyle{\sum_{k=1}^{30}}w_kS_{t_1}^k-L,0\big\},\text{ with } w_1,\dots,w_{30},L\in \R_+\bigg\}.
\end{align*}
We train the neural network only on $8000$ different samples, where each sample consists of  bid and ask prices of call options for $20$ different strikes of $30$ different underlying securities from data observed in June $2020$ with randomly generated weights $w_k$, $k=1,\dots,30$, and randomly generated strikes $L$, i.e., one single sample $\B{X_i}$ consists of $1861$ entries. These entries consist of $30 \cdot 20$ strikes of call options, $30 \cdot 20$ bid prices of call options, $30 \cdot 20$ ask prices of call options, as well as $30$ spot prices and $1$ strike $L$ of the basket option and $30$ associated weights $w_k$, $k=1,\dots,30$.

The corresponding prices are computed with the algorithm from \cite{neufeld2020model} which enables to compute precise price bounds even in this high-dimensional setting\footnote{With this algorithm it is even possible to compute price bounds of basket options that depend on $60$ securities.}.

After having trained the neural network, we test on $2000$ samples from data on options on the $S\& P$ $500$ that were observed in August $2020$, as described in Section~\ref{sec_testdata}. 
%As each sample consists of $1861$ real valued entries, we need to consider larger neural networks in comparison with the neural networks from Example~\ref{exa_call_options} and Example~\ref{exa_real_example_strategies}. Thus,  we increase the number of neurons in each hidden layer to $2048$.
We test only for the lower bound of the basket option and achieve a relative\footnote{ Note that here, as in Example~\ref{exa_real_example_strategies}, the relative error refers to the error after division with the \emph{weighted} sum of the spot prices, where the weights are determined by the weights in the payoff of the basket option. \black } mean absolute error of ${0.0101}$  and a relative mean squared error of $0.0002$ \black on the test set. Compare Figure~\ref{fig_accuracy_many_options_relative}, where we depict the relative error, i.e., we divide predictions and prices by the weighted sum of the spot prices.

\begin{figure}[h!]
\begin{center}
\includegraphics[scale=0.45]{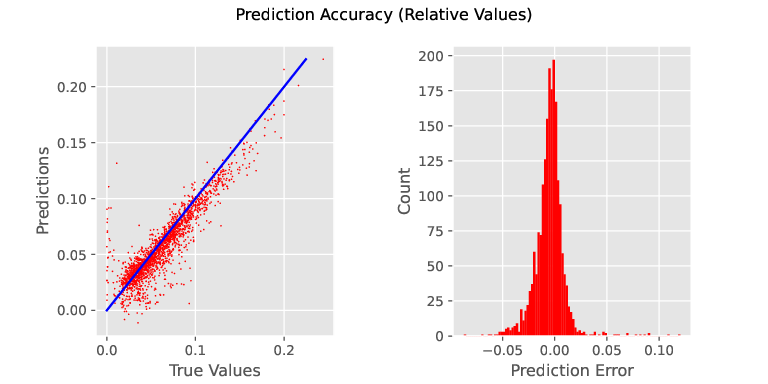}
\caption{This figure shows the relative error when predicting the lower price bound of a basket option that depends on $30$ underlying securities. We divide the target prices (and predicted prices) by the weighted spot prices, where the weights are the ones in the payoff function under consideration.}\label{fig_accuracy_many_options_relative}
\end{center}
\end{figure}
\end{exa}

\begin{rem}\label{rem_approach_in_practice}
\begin{itemize}
\item[(a)]
It turns out that indeed our proposed approach can be executed significantly faster than comparable methods that can be applied to compute model-free price bounds. The computation of $100$ price bounds in the setting of Example~\ref{exa_basket} takes $225.31$ seconds on a standard computer\footnote{\label{footnote_computational_time}  We used for the computations a \emph{Gen Intel(R) Core(TM) i7-1165G7, 2.80 GHz} processor with $40$ GB RAM. \black } when using the LSIP approach\footnote{ The \emph{Matlab}-code for the execution of the LSIP approach is provided under  \href{https://github.com/qikunxiang/ModelFreePriceBounds}{https://github.com/qikunxiang/ModelFreePriceBounds}. \black } from \cite{neufeld2020model}. The execution of a trained neural network to predict $100$ price bounds however only takes $ 0.00303$ seconds.
The execution of the neural network is therefore approximately $75000$ times faster.  This highlights the computational advantage of our proposed approach over comparable numerical methods, and further indicates that our approach indeed allows almost in \emph{real time} the model-free valuation of financial derivatives. { 
\item[(b)] We report quite fast training times that are required to train the neural networks. The computationally most expensive setting considered in Example~\ref{exa_basket} involves, after applying early stopping, the training of $284$ Epochs on $7200$ samples with a batch size of $256$. This training procedure takes in total only $142.88$ seconds\footnote{See footnote~\ref{footnote_computational_time},}.  }
\item[(c)] Note that even though the underlying
 approach is presented in a great generality, it can be easily modified to meet the potentially more specific requirements of applicants. 

If less call or put options are traded than the trained neural network contains, then one can simply set several strikes and prices to the same value and then apply the trained neural network on the smaller financial market.
Moreover, if the payoff function considered in the trained neural network depends on more assets than the same-type payoff function which one wants to price, then this is possible within our approach by adjusting the parameters from the more general payoff function. For example it is possible to determine the price bounds of call options after having trained price bounds of basket options as it was shown in Example~\ref{exa_basket_to_call}, compare also the overview provided in Table~\ref{tbl_payoff_functions} which clarifies which payoffs are of the same type.
Moreover, if one wants to take into account asset-specific investment constraints, both universal bounds $B$ and $\mathfrak{B}$ can be replaced by asset-specific bounds  $(B_i)_{i=1,\dots,n} $ and  option specific bounds $(\mathfrak{B}_{i,j,k})_{i=1,\dots,n, k=1,\dots,d, j=1,\dots,n_{ik}^{\operatorname{opt}}}$. The assertion of Theorem~\ref{thm_convergence} remains valid as the continuity mainly relies on a compactness argument which still can be applied.
\item[(d)] One major implicit assumptions of the presented approach is that the considered call and put options can be traded liquidly. Even though this assumption is usually fulfilled in practice one should verify this assumption carefully. Moreover, given that other traded options are considered sufficiently liquid, the presented approach can be extended in a straightforward way by including these options in \eqref{eq_phi} and \eqref{eq_definition_cost_functional}.
\item[(e)] If one is not interested in imposing a proper trading restriction through the bound $\mathfrak{B}$, then setting $\mathfrak{B}$ to a sufficiently large value does lead in practice to the same price bounds as in an unbounded setting. Therefore, also the optimal parameters of the neural networks that approximate these bounds are the same as in an unbounded setting. This holds true since both the trained parameters of the neural network and the corresponding trained trading strategy a posteriori turn out to remain bounded over the whole training period, compare e.g. \cite[Fig. 4]{baes2019low}.
\item[(f)] It is noteworthy that the presented approach allows to determine model-free price bounds of the most common types of traded financial derivatives by only training a couple of neural networks (one for each type of payoff function), compare the non-exhaustive Table~\ref{tbl_payoff_functions} for an overview which relies partly on the presentation provided in \cite{neufeld2020model}.  Table~\ref{tbl_payoff_functions} shows in particular which payoff functions are of the same type and can therefore be trained with a single neural network. { Note that with Example~\ref{exa_basket_to_call} we provide empirical evidence that this approach not only works in theory but indeed in practice. Recall that in Example~\ref{exa_basket_to_call}} we trained a neural network to predict price bounds of call options and basket options that are both of the same type.
\end{itemize}
\end{rem}
\begin{table}[h!]
 
\begin{center}{
\resizebox{\textwidth}{!}{
\begin{tabular}{llll} \toprule
     Type    & Name     & Payoff & Parameters $\Theta$ \\ \midrule
     I.  &Basket call option with weights $(w^k)_{k=1,\dots,d}$  and strike $L$ &$\max\{\sum_{k=1}^d w^k S_{t_1}^k-L,0\}$ & $\Theta = \{w^1,\dots,w^d \in \R,~ L \in \R\}$ \\   
      I.  &Basket put option with weights $(w^k)_{k=1,\dots,d}$  and strike $L$ &$\max\{L-\sum_{k=1}^d w^k S_{t_1}^k,0\}$ & $\Theta = \{w^1,\dots,w^d \in \R,~ L \in \R\}$ \\        
I.  &Call option on the $i$-th asset with strike $L$ &$\max\{S_{t_1}^i-L,0\}$ & $\Theta = \{L \in \R\}$ \\
I.        &Put option on the $i$-th asset with strike  $L$ &$\max\{L-S_{t_1}^i,0\}$ & $\Theta = \{L \in \R\}$  \\
I.          &Spread call options with weights $w^i,w^j$  and strike $L$& $\max\{w^iS_{t_1}^i-w^jS_{t_1}^j-L,0\}$ &$\Theta = \{ w^i, w^j \in \R, L \in \R \}$ \\         
I.        &Spread put options with weights $w_i,w_j$  and strike $L$& $\max\{L- (w^iS_{t_1}^i-w^jS_{t_1}^j),0\}$ &$\Theta = \{ w^i, w^j \in \R, L \in \R \}$  \\
II.        &Call-on-max with strike $L$& $\max\{ \max\{S_{t_1}^k,k=1,\dots,d\}-L,0\}$ &$\Theta = \{L \in \R \}$\\
II.        &Put-on-max with strike $L$& $\max\{L-\max\{S_{t_1}^k,k=1,\dots,d\},0\}$ &$\Theta = \{L \in \R \}$\\
III.        &Call-on-min with strike $L$& $\max\{ \min\{S_{t_1}^k,k=1,\dots,d\}-L,0\}$ &$\Theta = \{L \in \R \}$\\
III.        &Put-on-min with strike $L$& $\max\{L-\min\{S_{t_1}^k,k=1,\dots,d\},0\}$ &$\Theta = \{L \in \R \}$\\
IV.        &Best-of-calls option with strikes $L_1,\dots,L_d$& $\max\left\{ \max\{S_{t_1}^k-L_k,0\},k=1,\dots,d\right\}$ &$\Theta =  \{L_1,\dots,L_d \in \R \}$
\\
IV.        &Best-of-puts option with strikes $L_1,\dots,L_d$& $\max\left\{ \max\{L_k-S_{t_1}^k,0\},k=1,\dots,d\right\}$ &$\Theta = \{L_1,\dots,L_d \in \R \}$\\
 \bottomrule
\end{tabular}}}
\end{center}
\caption{The Table depicts the most common types of financial derivatives that can be trained by the presented approach. The leftmost column identifies the type of the respective financial derivative, where the price bounds of payoffs of the same type can be learned from a single neural network.}\label{tbl_payoff_functions}
\end{table}
\black 
\section{Martingale optimal transport}\label{sec_mot}
The presented approach from Section~\ref{sec_training} can easily be adjusted to modified market settings given that it is possible to establish a continuous relationship between the prevailing market scenario and resultant model-free price bounds of derivatives. In this section we show how the approach can be adapted to the setting used in martingale optimal transport (compare among many other relevant articles \cite{beiglbock2013model}, \cite{beiglbock2016problem}, \cite{cheridito2020martingale}, and \cite{dolinsky2014martingale}), where instead of observing a finite amount of call and put option prices, one assumes that the entire one-dimensional marginal distributions of the underlying assets at future dates are known. This situation is according to the Breeden-Litzenberger result~\cite{breeden1978prices} equivalent to the case where one can observe call and put option prices for a continuum of strikes on each of the associated maturities on which the financial derivative $\Phi\in B(\R_+^{nd},\R)$ depends.
In the martingale optimal transport case, one wants to compute the arbitrage-free upper price bound\footnote{We implicitly assume absence of a bid-ask spread and of transaction costs.} of $\Phi$ which leads to the maximization problem
\begin{equation}\label{eq_sup_mot}
\sup_{\Q \in \mathcal{M}(\B{\mu}_1,\dots,\B{\mu}_n)}\E_\Q[\Phi(S)],
\end{equation}
where\footnote{$\mathcal{P}(\R_+^{nd})$ denotes the set of all Borel probability measures on $\R_+^{nd}$.}
\begingroup\makeatletter\def\f@size{9}\check@mathfonts
\def\maketag@@@#1{\hbox{\m@th\normalsize\normalfont#1}}
\begin{equation}
\label{eq_definition_mot_set}
\begin{aligned}
\mathcal{M}(\B{\mu}_1,\dots,\B{\mu}_n):=\bigg\{\Q \in \mathcal{P}(\R_+^{nd})~\bigg|~&\Q\circ \left(S_{t_i}^k\right)^{-1} = \mu_i^k,\E_\Q[S_{t_{i+1}}^k~|~\boldsymbol{S}_{t_i},\dots,\boldsymbol{S}_{t_1}]=S_{t_i}^k  ~\Q\text{-a.s. for all } i,k\bigg\}
\end{aligned}
\end{equation}
\endgroup
describes the set of all $n$-step martingale measures with fixed one-dimensional marginals $\B{\mu}_i=(\mu_i^1,\dots,\mu_i^d)$, $i=1,\dots,n$, of all involved securities.

We show that, when considering a single asset and two future maturities, \eqref{eq_sup_mot} can be approximated by a properly constructed neural network. To that end, let
\[
\mathcal{P}_1(\R_+):=\left\{\Q\in \mathcal{P}(\R_+)~\middle|~\int_{\R_+} x \, \D \Q(x) < \infty\right\}
\]
 denote the set of probability measures on $\R_+$ with existing first moment.
Further, we introduce the $1$-Wasserstein-distance $\operatorname{W}(\cdot,\cdot)$ between two measures $\mu_1,\mu_2 \in \mathcal{P}_1(\R_+)$, which is defined through 
\[
\operatorname{W}(\mu_1,\mu_2) := \inf_{\pi \in \Pi(\mu_1,\mu_2)}\int_{\R_+^2} | u-v |\,\D\pi(u,v),
\]
with $\Pi(\mu_1,\mu_2)$ denoting the set of all couplings\footnote{More precisely, the set  $\Pi(\mu_1,\mu_2)$ is defined as $\Pi(\mu_1,\mu_2):=\left\{ \pi \in \mathcal{P}_1(\R_+^2)~:~\pi \circ S_{t_i}^{-1}=\mu_i,~i=1,2\right\}$} of $\mu_1$ and $\mu_2$, compare e.g. \cite{villani2008optimal}. 

We recall the construction of the $\mathcal{U}$-quantization from \cite{baker2012martingales}. Given some probability measure $\mu \in \mathcal{P}_1(\R_+)$ and some $N \in \N$ we set for $i=1,\dots,N$
\[
x_i^{(N)}(\mu) := N \int_{(i-1)/N}^{i/N} F_{\mu}^{-1}(u)\,\D u,
\]
where $F_{\mu}^{-1}(u):= \inf \{ x\in \R_+ :   F_{\mu}(x) \geq u \}$ denotes the $u$-quantile associated to the cumulative distribution function $F_{\mu}$ of $\mu$, and we denote $\B{x}^{(N)}(\mu):=\left(x_l^{(N)}(\mu)\right)_{l=1,\dots,N}$. Then, by \cite[Theorem 2.4.12.]{baker2012martingales} it holds that
\begin{equation}\label{eq_u_quant}
\mathcal{U}^{(N)}(\mu):= \frac{1}{N}\sum_{i=1}^N \delta_{x_i^{(N)}(\mu)}
\end{equation}
converges weakly to $\mu$ for $N \rightarrow \infty$. Since the mean of $\mathcal{U}^{(N)}(\mu)$ and $\mu$ coincide for all $N \in \N$ due to \cite[Lemma 2.4.4.]{baker2012martingales} and $\mu$, $\mathcal{U}^{(N)}(\mu)\in \mathcal{P}_1(\R_+)$, 
we further obtain convergence in the $1$-Wasserstein-distance, compare \cite[Definition 6.8]{villani2008optimal}. 
This means particularly that  for all $\mu \in \mathcal{P}_1(\R_+)$ and for all $\delta>0$ there exists some $N\in \N$ such that $\operatorname{W}(\mathcal{U}^{(N)}(\mu),\mu)< \delta$.

We derive the following  novel \black result which asserts  for the first time \black  that two-marginal martingale optimal transport problems can be approximated arbitrarily well by neural networks.\footnote{We recall that $\|\cdot \|_2$ is an arbitrary norm on $\R^2$.} 
\begin{thm}\label{lem_approx_mot}
Let $\Phi: \R_+^2 \rightarrow \R$ be continuous such that $\sup_{x_1,x_2 \in \R_+}\frac{|\Phi(x_1,x_2)|}{1+x_1+x_2} < \infty $.
Then, for all $\varepsilon >0$, $N \in \N$, and compact sets $\mathbb{K} \subset \R_+^{N}$, there exists a neural network $\mathcal{N} \in \mathfrak{N}_{2N,2}$ such that for all $(\mu_1,\mu_2) \in \mathcal{P}_1(\R_+)\times \mathcal{P}_1(\R_+)$ with $\mu_1 \preceq \mu_2$\footnote{Here $\preceq$ denotes the convex order for measures $(\mu_1,\mu_2) \in \mathcal{P}_1(\R_+)\times \mathcal{P}_1(\R_+)$, i.e., $\mu_1 \preceq \mu_2$ means \\ $\int_{\R_+} f \D \mu_1 \leq \int_{\R_+} f \D \mu_2$ for all convex functions $f:\R_+ \rightarrow \R$.}, there exists some $\delta >0$ such that if  $\operatorname{W}(\mathcal{U}^{(N)}(\mu_1),\mu_1)<\delta$, $\operatorname{W}(\mathcal{U}^{(N)}(\mu_2),\mu_2)<\delta$, and $\B{x}^{(N)}(\mu_1),\B{x}^{(N)}(\mu_2) \in \mathbb{K}$, then
\begin{equation}\label{eq_eps_expression}
\begin{aligned}
\bigg\|&\mathcal{N}  \left(\B{x}^{(N)}(\mu_1),\B{x}^{(N)}(\mu_2)\right)-\left(\inf_{\Q \in \mathcal{M}(\mu_1,\mu_2)}\E_\Q[\Phi],\sup_{\Q \in \mathcal{M}(\mu_1,\mu_2)}\E_\Q[\Phi]\right)\bigg\|_2<\varepsilon.
\end{aligned}
\end{equation}
\end{thm}
\begin{proof}
See Section \ref{appendix_proof}.
\end{proof}
\begin{rem}
The assumption in Theorem~\ref{lem_approx_mot} stating that the atoms of the approximating $\mathcal{U}$-quantizations are contained in some prespecified compact set $\mathbb{K}$ can always be fulfilled  if one only considers marginals with support in $\mathbb{K}$.

If one does not want to restrict to compactly supported marginals, one can start with $\widehat{\mu}_1, \widehat{\mu}_2\in \mathcal{P}_1(\R_+)$ and consider for every $r>0$ the sets $B_r(\widehat{\mu}_i):=\left\{\mu_i \in \mathcal{P}_1(\R_+)~:~W(\mu_i, \widehat{\mu}_i) \leq r\right\}$ for $i=1,2$. Then there exists some compact set $\mathbb{K} \subset \R_+^N$ s.t. $x_l^{(N)}(\mu_i) \in \mathbb{K}$ for all $l=1,\dots,N$, $\mu_i \in B_r(\widehat{\mu}_i)$, $i=1,2$.

Indeed, there exists some constant $C>0$ such that 
\begin{equation}\label{eq_compact_K}
\left|x_l^{(N)}(\widehat{\mu}_i)\right|\leq  C \text{ for all } l=1,\dots,N,i=1,2.
\end{equation}
Moreover, for all $\mu_i \in B_r(\hat{\mu}_i)$, $i=1,2$, we have for all $l=1,
\dots,N$ that
\begingroup\makeatletter\def\f@size{9}\check@mathfonts
\def\maketag@@@#1{\hbox{\m@th\normalsize\normalfont#1}}
\begin{align}
\left|x_l^{(N)}(\widehat{\mu}_i)-x_l^{(N)}({\mu_i})\right|=\left|N \int_{(l-1)/N}^{l/N} F_{\widehat{\mu}_i}^{-1}(u)-F_{{\mu}_i}^{-1}(u)\,\D u\right| \notag &\leq N \int_{(l-1)/N}^{l/N} \left|F_{\widehat{\mu}_i}^{-1}(u)-F_{{\mu_i}}^{-1}(u)\right|\D u \notag \\
&\leq  N \int_{0}^{1} \left|F_{\widehat{\mu}_i}^{-1}(u)-F_{{\mu_i}}^{-1}(u)\right|\D u \notag \\
&= N \cdot {W}({\widehat{\mu}_i},{\mu_i}) \leq N \cdot r, \label{eq_equality_fh}
\end{align}
\endgroup
see for the last equality in \eqref{eq_equality_fh} also \cite[Equation 3.1.6]{rachev1998mass} and \cite[Equation 3.5]{ruschendorf2007monge}. This implies for all $i=1,2$ that \\$\left(x_1^{(N)}(\mu_i),\dots,x_N^{(N)}(\mu_i)\right) \in \mathbb{K}:=\left\{ (x_1,\dots,x_N) \in \R_+^N~|~|x_j| \leq C+ N\cdot r \text{ for all } j \right\}$.
\end{rem}

\begin{rem}
If we are only interested in predicting the upper bound $\sup_{\Q \in \mathcal{M}(\mu_1,\mu_2)}\E_\Q[\Phi]$, then one can see, by carefully reading the proof of Theorem~\ref{lem_approx_mot}, that it suffices to assume that $\Phi$ is upper semi-continuous (and linearly bounded) to obtain the existence of a neural network $\mathcal{N} \in \mathfrak{N}_{2N,1}$ that approximates the bound as in \eqref{eq_eps_expression} w.r.t. $\|\cdot \|_1$. Similarly, to derive the existence of a neural network approximating the lower bound $\inf_{\Q \in \mathcal{M}(\mu_1,\mu_2)}\E_\Q[\Phi]$ it is only necessary to assume that $\Phi$ is lower semi-continuous (and linearly bounded).
\end{rem}
In the case with two marginals, the training routine from Algorithm~\ref{algo_training_nn} modifies to the below presented Algorithm~\ref{algo_training_mot} which is also depicted in Figure~\ref{fig_diagram_mot}.

\begin{algorithm}\label{algo_training_mot}
\SetAlgoLined
\small{
\SetKwInOut{Input}{Input}
\SetKwInOut{Output}{Output}
\Input{Payoff function $\Phi:\R_+^{2} \rightarrow \R$ fulfilling the assumptions of Theorem~\ref{lem_approx_mot}; \newline $N$ describing the number of maximal supporting values of the approximating distribution; Number of samples $\mathcal{S}$; Method to create samples such that the associated marginal distributions increase in convex order (cf. Remark~\ref{rem_mot_algo}~(a));}

\For {$i$ in $\{1,\dots,\mathcal{S}\}$} {Create samples $\B{X_i} = (x_1^i,\dots,x_N^i,y_1^i,\dots,y_N^i) \in \R_+^{2N}$ such that $\frac{1}{N}\sum_{j=1}^N \delta_{x_j^i} \preceq \frac{1}{N}\sum_{j=1}^N \delta_{y_j^i}$; \\
Compute, e.g. via linear programming (compare \cite{guo2019computational}), the target value 
\footnotesize{\begin{align*}
\B{Y_i} = \Bigg(&\inf_{\Q \in \mathcal{M}\left(\frac{1}{N}\sum_{j=1}^N \delta_{x_j^i},~\frac{1}{N}\sum_{j=1}^N \delta_{y_j^i}\right)}\E_\Q[\Phi],\sup_{\Q \in \mathcal{M}\left(\frac{1}{N}\sum_{j=1}^N \delta_{x_j^i},~\frac{1}{N}\sum_{j=1}^N \delta_{y_j^i}\right)}\E_\Q[\Phi]\Bigg)\in \R^2;
\end{align*}}}

Train via back-propagation a neural network $\mathcal{N} \in \mathfrak{N}_{2N,2}$ to predict $\B{Y_i}$ given $\B{X_i}$, i.e., such that $\mathcal{N}(\B{X_i}) \approx \B{Y_i}$ (which is possible due to Theorem~\ref{lem_approx_mot}); \\
\Output{Trained neural network $\mathcal{N} \in \mathfrak{N}_{2N,2}$;}
\caption{Training of a neural network for the computation of the price bounds of some financial derivative $\Phi$ in the MOT setting}
}
\end{algorithm}
\begin{rem}~\label{rem_mot_algo}
\begin{itemize}
\item[(a)]
The critical point in Algorithm~\ref{algo_training_mot} is the method which is employed to create discrete samples of marginals that are increasing in convex order. It is a priori not obvious how to create a good sample set. One working methodology is to draw samples from marginals that are \emph{similar} to marginals one wants to predict, then to apply $\mathcal{U}$-quantization and to discard measures that do not increase in convex order. Compare also Example ~\ref{exa_mot_wo_constraints}. \emph{Being similar} can for example mean to be close in Wasserstein distance or being from the same parametric family of probability distributions.
\item[(b)] Since the proof of Theorem~\ref{lem_approx_mot} relies on the continuity of the MOT problem, as stated in \cite{backhoff2019stability}, \cite{neufeld2021stability}, and \cite{wiesel2019continuity}, for which - at the moment - no extension to the case with more than two marginals or in multidimensional settings is known, we stick to the two marginal case in the one-dimensional setting.
\item[(c)] The approach can be extended to the case with information on the variance as in \cite{lutkebohmert2019tightening}, with Markovian assumptions as imposed in \cite{eckstein2019martingale} and \cite{sester2020robust}, or even more general constraints on the distribution (see e.g. \cite{ansari2020improved}), whenever it is possible to establish a continuous relationship between marginals and prices. Compare also Example~\ref{exa_mot_varaince}, where we consider a constrained martingale optimal transport problem.
\end{itemize}

\end{rem}

To compute the robust bounds $\inf_{\Q \in \mathcal{M}(\mu_1,\mu_2)}\E_\Q[\Phi]$ and $\sup_{\Q \in \mathcal{M}(\mu_1,\mu_2)}\E_\Q[\Phi]$ for arbitrary (possibly continuous) marginal distributions  $\mu_1 \preceq \mu_2$, with the above explained approach, we approximate $\mu_1$ and $\mu_2$ in the Wasserstein distance by $\mathcal{U}^{(N)}(\mu_1)\preceq \mathcal{U}^{(N)}(\mu_2)$  and then compute the resultant price bound via 
$\mathcal{N}(\B{x}^{(N)}(\mu_1),\B{x}^{(N)}(\mu_2))$, where $\mathcal{N} \in \mathfrak{N}_{2N,2}$ denotes the neural network which was trained according to Algorithm~\ref{algo_training_mot}. 
\subsection{Examples}
\label{sec_mot_exa}
We train a neural network according to Algorithm~\ref{algo_training_mot}. Thus, the input features are, in contrast to the examples from Section~\ref{sec_market_data}, not directly option prices but are given by marginal distributions which are discretized according to the $\mathcal{U}$-quantization.
In the following we fix a payoff function $\Phi(S_{t_1},S_{t_2})=|S_{t_1}-S_{t_2}|$ and present the results when predicting
$\inf_{\Q \in \mathcal{M}(\mu_1,\mu_2)}\E_\Q[\Phi(S_{t_1},S_{t_2})]$ and $\sup_{\Q \in \mathcal{M}(\mu_1,\mu_2)}\E_\Q[\Phi(S_{t_1},S_{t_2})]$ via neural networks in dependence of the marginals $\mu_1$ and $\mu_2$.
\subsubsection{Implementation}
To train the neural network we create according to Algorithm~\ref{algo_training_mot} numerous artificial marginals with a fixed number $N \in \N$ of supporting values. In the examples below we choose $N =20$. To discretize continuous marginal distributions we apply the $\mathcal{U}$-quantization as introduced in \eqref{eq_u_quant}. 
Below, we describe which marginal distributions $\mu_1^i,\mu_2^i$ we generate for $i=1,\dots,\mathcal{S}$ with $\mathcal{S}$ being the total number of samples.
\begin{enumerate}
\item \textbf{Log-Normally distributed marginals} \\(if $i\mod 4 = 0$)\\
$\mu_1^i \sim \mathcal{L}\mathcal{N}(\mu^i-(\sigma_1^i)^2/2,(\sigma_1^i)^2)$ with $\mu^i \sim \mathcal{U}([-2,2])$ and $\sigma_1^i \sim \mathcal{U}([0,0.5])$, \\
$\mu_2^i \sim 
\mathcal{L}\mathcal{N}(\mu^i-(\sigma_2^i)^2/2,(\sigma_2^i)^2)$ with $\sigma_2^i = \sigma_1^i \cdot \widetilde{\sigma_2}^i$, where $\widetilde{\sigma_2}^i \sim \mathcal{U}([1,2])$;
\item \textbf{Uniform marginals} \\(if $i\mod 4 = 1$)\\
 $\mu_1^i \sim \mathcal{U}([\mu^i-a^i,\mu^i+a^i])$,\\
 $\mu_2^i \sim \mathcal{U}([\mu^i-b^i,\mu^i+b^i])$ with $\mu^i \sim \mathcal{U}([10,20]),~a^i\sim \mathcal{U}([0,5]),~b^i\sim \mathcal{U}([a^i,a^i+5])$;
 \item \textbf{Continuous and discrete uniform marginals} \\(if $i\mod 4 = 2$)\\
 $\mu_1^i \sim \mathcal{U}([\mu^i-a^i,\mu^i+a^i])$,\\
 $\mu_2^i \sim \mathcal{U}(\{\mu-a^i,\mu+a^i\})$ with $\mu^i \sim \mathcal{U}([5,10]),~a^i\sim \mathcal{U}([0,5])$;
\item \textbf{Uniform and triangular marginals} \\(if $i\mod 4 = 3$)\\
 $\mu_1^i \sim \mathcal{U}([m^i-l^i/2,m^i+l^i/2])$ for $l^i \sim \mathcal{U}([0,5])$, $m^i \sim \mathcal{U}([l^i,l^i+10])$ and triangular marginals $\mu_2^i$ with lower limit $l^i$, upper limit  $u^i \sim \mathcal{U}([m^i,m^i+10])$ and mode $m^i$. 
\end{enumerate}
If the generated marginals are in convex order\footnote{This can be easily checked by verifying that $\int_{\R_+}(x-L)_+\D\mu_1(x) \leq \int_{\R_+} (x-L)_+\D\mu_2(x)$ for all atoms $L$ as well as $\int_{\R_+} x\D\mu_1(x) = \int_{\R_+} x\D\mu_2(x)$, compare e.g. \cite{ma2000convex}.}, then we add the discretized values $\mathcal{U}^{(N)}(\mu_1^i)$, $\mathcal{U}^{(N)}(\mu_2^i)$ to the sample set $(\B{X_i})_{i=1,\dots,\mathcal{S}}$ and compute via linear programming $\inf_{\Q \in \mathcal{M}(\mu_1^i,\mu_2^i)}\E_\Q[\Phi]$ as well as $\sup_{\Q \in \mathcal{M}(\mu_1^i,\mu_2^i)}\E_\Q[\Phi]$ which we then add as corresponding target values to $(\B{Y_i})_{i=1,\dots,\mathcal{S}}$.

\subsubsection{Architecture of the neural networks}
Given a set of samples $(\B{X_i})_{i=1,\dots,\mathcal{S}}$ and a set of targets $(\B{Y_i})_{i=1,\dots,\mathcal{S}}$ we train a neural network using the back-propagation algorithm with an Adam optimizer implemented in \emph{Python} using \emph{Tensorflow} similar to Section~\ref{sec_market_data}. The loss function is a $L^2$-loss function, the neural networks comprise $3$ hidden layers with $512$ neurons each, and with \emph{ReLU} activation functions. \\
For further details of the code we refer to \href{https://github.com/juliansester/deep\_model\_free\_pricing}{https://github.com/juliansester/deep\_model\_free\_pricing}.

\begin{exa}[MOT without constraints]\label{exa_mot_wo_constraints}
We report the results for a neural network trained to $\mathcal{S}=100,000$ samples that are generated according to the procedure described above.
We split the samples in training set, test set ($10\%$ of the samples) and validation set ($20\%$ of the training samples) and report a mean absolute error of ${0.0082}$  as well as mean squared error of $0.0001$ \black on the test set after $1000$ epochs of training with early stopping.
The accuracy of the trained neural network on the test set is displayed in Figure~\ref{fig_accuracy_mot}.
Additionally, we test the neural network in the following specific situations.
\begin{itemize}
\item[(a)] $\mu_1 =\mathcal{LN}(0.5-0.5\cdot 0.25^2,0.25^2)$,\\$\mu_2 = \mathcal{LN}(0.5-0.5\cdot 0.5^2,0.5^2)$
\item[(b)] $\mu_1 =\mathcal{LN}(-0.05,0.1),~\mu_2 = \mathcal{LN}(-0.1,0.2)$
\item[(c)] $\mu_1=\mathcal{U}([8,12]),~\mu_2=\mathcal{U}([5,15])$
\item[(d)] $\mu_1=\mathcal{U}([5,10]),~\mu_2=\mathcal{U}(\{5,10\})$
\item[(e)] $\mu_1=\mathcal{U}([2,4])$,\\
$ \frac{\D \mu_2}{\D \lambda}(x)= (x-1)/3\one_{[1,2]}(x)+ (1/3)\one_{[2,4]}(x)+(5-x)/3\one_{[4,5]}(x)$, where $\lambda$ denotes the Lebesgue-measure.
\end{itemize}
The results are displayed in Table~\ref{tbl_nn_results} and indicate that the bounds can indeed be approximated very precisely.
\begin{table}[h!]
\begin{center}

\caption{The table displays for different marginals the approximated lower bounds and upper bounds that are computed via trained neural networks~(NN) and via a linear programming approach~(LP).}\label{tbl_nn_results}
\resizebox{0.8\textwidth}{!}
{\begin{tabular}{@{}l
 c c c c | c}
 \toprule
 &Lower bound (LP)    &Lower bound (NN)    &Upper bound (LP)   &Upper bound (NN) & Cumulative Error\\
\midrule
(a) &$0.2363$&$0.2573$&$0.4226$&$0.4210$&$0.0226$\\
(b) &$0.0814$&$0.0939$&$0.1870$&$0.1946$&$0.0202$\\
(c) &$1.7491$&$1.7503$&$2.6220$&$2.6082$&$0.0150$\\
(d) &$1.6688$&$1.6659$&$1.6687$&$1.6636$&$0.0080$\\
(e) &$0.3587$&$0.3626$&$0.7215$&$0.7151$&$0.0103$\\

\bottomrule
\end{tabular}}
\end{center}
\end{table}

\begin{figure}[h!]
\begin{center}
\includegraphics[scale=0.45]{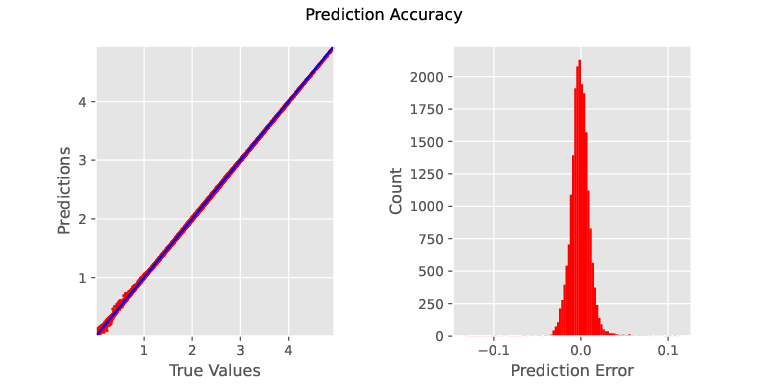}
\caption{This figure illustrates the accuracy of the predictions of the trained neural network in the MOT setting, evaluated on the test set.}\label{fig_accuracy_mot}
\end{center}
\end{figure}
\end{exa}

In the next example we show how the introduced methodology can be applied to constrained martingale optimal transport problems.

\begin{exa}[MOT with variance constraints]\label{exa_mot_varaince}
We train the neural network with the same artificial samples of marginals as in the previous Example~\ref{exa_mot_wo_constraints}. In addition to the discretized marginals $\mathcal{U}^{(N)}(\mu_1^i)$, $\mathcal{U}^{(N)}(\mu_2^i)$, we consider as an additional feature a pre-specified level of the variance of the returns as in \cite{lutkebohmert2019tightening}. This means we only consider martingale measures $\Q$ which additionally fulfil $\operatorname{Var}_\Q(S_{t_2}/S_{t_1})=\sigma_{12}^2$. In particular, $\sigma_{12}$ is thus, an additional feature of the samples $\B{X_i}$. 
It was already indicated in Remark~\ref{rem_mot_computation}~(c), that the approximation through neural networks, as stated in Theorem~\ref{lem_approx_mot}, can also be obtained in this case due to \cite[Theorem A.6.]{lutkebohmert2019tightening} ensuring continuity of the map 
\begin{align*}
(\mu_1,\mu_2)\mapsto \sup \big\{\E_{\Q}[\Phi]~\big|~&\Q \in \mathcal{M}(\mu_1,\mu_2) \text{ and } \operatorname{Var}_\Q(S_{t_2}/S_{t_1})=\sigma_{12}^2 \big\}
\end{align*}
 when $\Phi$ is Lipschitz-continuous and when the marginals have compact support. Then, an adaption of Theorem~\ref{lem_approx_mot} is straightforward.
The results of a neural network approximation for the payoff function $\Phi(S_{t_1},S_{t_2})=|S_{t_2}-S_{t_1}|$ for the marginal distributions (c) and (e) from Example~\ref{exa_mot_wo_constraints} are displayed in Figure~\ref{fig_variance_plots}. The marginals from (a) and (b) are omitted as they do not satisfy the conditions from \cite[Theorem A.6.]{lutkebohmert2019tightening}, whereas the marginals from (d) are omitted, since the value of lower and upper bound coincide and they therefore do not depend on a pre-specified level of the variance. The neural network was trained on $500 000$ samples for $1000$ epochs with early stopping. We assign $10\%$ of the samples to the test set and $20\%$ of the remaining training samples to the validation set (which is relevant for the early stopping rule). The resulting mean absolute error on the test set is ${0.0044}$,  the mean squared error is  $0.00003$\black. In Figure~\ref{fig_accuracy_variance} we show to which degree the predictions deviate from the target values on the test set implying that the accuracy of the predictions is indeed very high on the test set.
\begin{figure}[h!]
\begin{center}
\subfloat[\label{fig_variance_plots}]{
\includegraphics[scale=0.45]{inc_arxiv/variance_plots.png}
}
\subfloat[\label{fig_accuracy_variance}]{\includegraphics[scale=0.45]{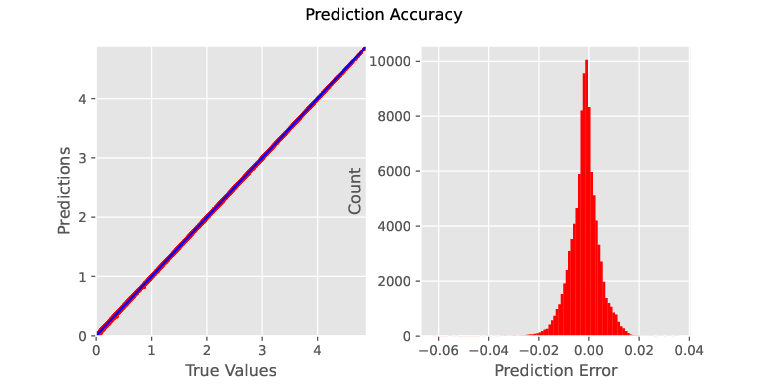}}

\caption{(a): This figure shows the accuracy of a neural network that was trained with $500 000$ samples. The accuracy is displayed for the test marginals from (c) and (e) of Example~\ref{exa_mot_wo_constraints}. The points indicate the upper and lower bounds for the prices under the influence of variance information obtained from the trained neural network (NN) in comparison with the precise bounds, computed with a linear programming (LP) approach, that are indicated by the solid lines.\\
(b): This figure illustrates the accuracy of the predictions of the trained neural network in the MOT setting with variance constraints, evaluated on the test set.}
\end{center}
\end{figure}
\end{exa}

\section{Proofs}\label{appendix_proof}
\begin{proof}[Proof of Theorem~\ref{thm_convergence}~(a)]
We prove the continuity of $\left(\boldsymbol{K}, \B{\pi}, \B{S_{t_0}}, \theta\right) \mapsto \overline{D}^{\mathfrak{B},B}_{(\boldsymbol{K},\boldsymbol{\pi},\B{S_{t_0}})}\left(\Phi_{\theta}\right)$.
 The continuity of $\left(\boldsymbol{K}, \B{\pi}, \B{S_{t_0}}, \theta\right) \mapsto \underline{D}^{\mathfrak{B},B}_{(\boldsymbol{K},\boldsymbol{\pi},\B{S_{t_0}})}\left(\Phi_{\theta}\right)$ can be obtained analogously.\\
Let $\varepsilon > 0$ and let $(\boldsymbol{K}, \B{\pi}, \B{S_{t_0}}, \theta) \in \mathbb{K}_1$. For any $\widetilde{\delta}>0$, by the continuity of 
\begin{align*}
\left(\Theta,d_p\right) &\rightarrow \left(C(\R_+^{nd},\R),d_{\infty,B}\right)\\
\theta &\mapsto \Phi_\theta,
\end{align*}
we can choose $\delta>0$ sufficiently small such that for all $\widetilde{\theta} \in \Theta$ we have that
\begin{equation}\label{ineq_lemma_delta}
\begin{aligned}
\|\theta-\widetilde{\theta}\|_p < \delta
\end{aligned}
\end{equation}
implies that
\begin{equation}\label{eq_max_ineq}
\begin{aligned}
\|&\Phi_{\theta}-\Phi_{\widetilde{\theta}}\|_{\infty,B} <\widetilde{\delta}.
\end{aligned}
\end{equation}
Moreover, we can choose $\widetilde{\delta}$ and $\delta$ small enough to ensure that
\begin{equation}\label{eq_f1df2deps}
\widetilde{\delta}+\delta (1+\kappa)  \mathfrak{B}  + {\delta}\mathfrak{B} < \varepsilon/2.
\end{equation}
We pick some $\delta >0, \widetilde{\delta}>0$ such that  the implication from  \eqref{ineq_lemma_delta} to \eqref{eq_max_ineq} is satisfied, and such that \eqref{eq_f1df2deps} holds \black true and let $(\widetilde{\boldsymbol{K}}, \widetilde{\B{\pi}}, \boldsymbol{\widetilde{S}_{t_0}}, \widetilde{\theta}) \in \mathbb{K}_1$ satisfy for all $i,j,k$ that
\begin{equation}\label{eq_assumption_delta}
\begin{aligned}
&\left|K_{ijk}^{\operatorname{call}}-\widetilde{K}_{ijk}^{\operatorname{call}}\right| < \delta, &&\left|K_{ijk}^{\operatorname{put}}-\widetilde{K}_{ijk}^{\operatorname{put}}\right| < \delta,\\
&\|\theta-\widetilde{\theta}\|_p < \delta,&&|S_{t_0}^k-\widetilde{S}_{t_0}^k| < \delta,\\
&|\pi_{\operatorname{call},i,j,k}^+- \widetilde{\pi}_{\operatorname{call},i,j,k}^+|<\delta,&&|\pi_{\operatorname{call},i,j,k}^-- \widetilde{\pi}_{\operatorname{put},i,j,k}^-|<\delta,\\
&|\pi_{\operatorname{put},i,j,k}^+- \widetilde{\pi}_{\operatorname{call},i,j,k}^+|<\delta,&&|\pi_{\operatorname{put},i,j,k}^-- \widetilde{\pi}_{\operatorname{put},i,j,k}^-|<\delta.
\end{aligned}
\end{equation}
First, assume that 
\begin{equation}
\overline{D}^{\mathfrak{B},B}_{(\boldsymbol{\widetilde{K}},\widetilde{\boldsymbol{\pi}},\B{\widetilde{S}_{t_0}})}\left(\Phi_{\widetilde{\theta}}\right) \geq \overline{D}^{\mathfrak{B},B}_{(\boldsymbol{K},\boldsymbol{\pi},\B{S_{t_0}})}\left(\Phi_{\theta}\right). \label{eq_greater_assumption}
\end{equation}
In this case, consider parameters $\widehat{a},(\B{\widehat{c}_{ijk}})_{i,j,k}$, $(\B{\widehat{p}_{ijk}})_{i,j,k}$, $(\widehat{\Delta}_i^k)_{i,k}$ such that $\Psi^{(\boldsymbol{K},\B{S_{t_0}})}_{(\widehat{a},\B{\widehat{c}_{ijk}},\B{\widehat{p}_{ijk}},\widehat{\Delta}_i^k)}(\B{s}) \geq \Phi_{\theta}(\B{s})$ for all $\B{s} \in [0,B]^{nd}$, $\Sigma\left({\B{\widehat{c}_{ijk}},\B{\widehat{p}_{ijk}},\widehat{\Delta}_i^k}\right) \leq \mathfrak{B}$, and such that 
\begin{equation}\label{eq_c_epsilon_half}
\C\left(\Psi^{(\boldsymbol{K},\B{S_{t_0}})}_{(\widehat{a},\B{\widehat{c}_{ijk}},\B{\widehat{p}_{ijk}},\widehat{\Delta}_i^k)},\boldsymbol{\pi}\right)\leq \overline{D}^{\mathfrak{B},B}_{(\boldsymbol{K},\boldsymbol{\pi},\B{S_{t_0}})}\left(\Phi_{\theta}\right)+\varepsilon/2,
\end{equation}
which is possible due to \eqref{eq_no_arbitrage}.
Then, we obtain by definition of the semi-static strategies defined in \eqref{eq_phi} and by \eqref{eq_assumption_delta} that
\begingroup\makeatletter\def\f@size{9}\check@mathfonts
\def\maketag@@@#1{\hbox{\m@th\normalsize\normalfont#1}}
\begin{align*}
\left| \Psi^{(\boldsymbol{K},\B{S_{t_0}})}_{(\widehat{a},\B{\widehat{c}_{ijk}},\B{\widehat{p}_{ijk}},\widehat{\Delta}_i^k)}(\boldsymbol{s}) -\Psi^{(\boldsymbol{\widetilde{K}},\B{\widetilde{S}_{t_0}})}_{(\widehat{a},\B{\widehat{c}_{ijk}},\B{\widehat{p}_{ijk}},\widehat{\Delta}_i^k)}(\boldsymbol{s}) \right| 
&=\Bigg|\sum_{i=1}^n\sum_{k=1}^d\sum_{j=1}^{n_{ik}^{\operatorname{opt}}} \bigg[(\widehat{c}_{ijk}^{~+}-\widehat{c}_{ijk}^{~-})\cdot\left(\max\left\{s_i^k-K_{ijk}^{\operatorname{call}},0\right\}-\max\left\{s_i^k-\widetilde{K}_{ijk}^{\operatorname{call}},0\right\}\right)\bigg]\\
&\hspace{0.5cm}+\sum_{i=1}^n\sum_{k=1}^d\sum_{j=1}^{n_{ik}^{\operatorname{opt}}} \bigg[(\widehat{p}_{ijk}^{~+}-{\widehat{p}_{ijk}}^{~-})\cdot\left(\max\left\{K_{ijk}^{\operatorname{put}}-s_i^k,0\right\}-\max\left\{\widetilde{K}_{ijk}^{\operatorname{put}}-s_i^k,0\right\}\right)\bigg] \\
&\hspace{0.5cm}+ \sum_{k=1}^d \bigg(\widehat{\Delta}_0^k(\widetilde{S}_{t_0}^k-{S}_{t_0}^k)+\kappa|\widehat{\Delta}_0^k|(|\widetilde{S}_{t_0}^k|-|S_{t_0}^k|)\bigg)\Bigg|\\
&\leq \delta (1+\kappa) \cdot \Sigma\left({\B{\widehat{c}_{ijk}},\B{\widehat{p}_{ijk}},\widehat{\Delta}_i^k}\right) \leq \delta (1+\kappa)  \mathfrak{B} 
\end{align*}
\endgroup
for all $\B{s}=(s_1^1,\dots,s_1^d,\dots,s_n^1,\dots,s_n^d)\in [0,B]^{nd}$.
Thus it holds pointwise on $[0,B]^{nd}$ that
\begin{equation}\label{eq_proof_f1d}
\begin{aligned}
\Psi^{(\boldsymbol{\widetilde{K}},\B{\widetilde{S}_{t_0}})}_{(\widetilde{\delta}+\delta (1+\kappa)  \mathfrak{B} +\widehat{a},\B{\widehat{c}_{ijk}},\B{\widehat{p}_{ijk}},\widehat{\Delta}_i^k)} &= \widetilde{\delta}+\delta (1+\kappa)  \mathfrak{B} + \Psi^{(\boldsymbol{\widetilde{K}},\B{\widetilde{S}_{t_0}})}_{(\widehat{a},\B{\widehat{c}_{ijk}},\B{\widehat{p}_{ijk}},\widehat{\Delta}_i^k)} \geq \widetilde{\delta}+\Psi^{(\boldsymbol{{K}},\boldsymbol{{S}}_{t_0})}_{(\widehat{a},\B{\widehat{c}_{ijk}},\B{\widehat{p}_{ijk}},\widehat{\Delta}_i^k)}  \geq \widetilde{\delta}+\Phi_{\theta}\geq \Phi_{\widetilde{\theta}},
\end{aligned}
\end{equation}
where the last inequality follows due to \eqref{eq_max_ineq}, since $\|\Phi_{\theta}-\Phi_{\widetilde{\theta}}\|_{\infty,B} < \widetilde{\delta}$.
This then yields by the definition of the cost function in \eqref{eq_definition_cost_functional}, by \eqref{eq_max_ineq}, \eqref{eq_f1df2deps}, and by \eqref{eq_assumption_delta} that
\begingroup\makeatletter\def\f@size{8.5}\check@mathfonts
\def\maketag@@@#1{\hbox{\m@th\normalsize\normalfont#1}}
\begin{equation}\label{eq_53a}
\begin{aligned}
\bigg|\C\left(\Psi^{(\boldsymbol{\widetilde{K}},\B{\widetilde{S}_{t_0}})}_{(\widetilde{\delta}+\delta (1+\kappa)  \mathfrak{B} +\widehat{a},\B{\widehat{c}_{ijk}},\B{\widehat{p}_{ijk}},\widehat{\Delta}_i^k)},\widetilde{\boldsymbol{\pi}}\right)-\C\left(\Psi^{(\boldsymbol{K},\B{S_{t_0}})}_{(\widehat{a},\B{\widehat{c}_{ijk}},\B{\widehat{p}_{ijk}},\widehat{\Delta}_i^k)} ,\boldsymbol{\pi} \right)\bigg| &\leq \widetilde{\delta}+\delta (1+\kappa)  \mathfrak{B} +{\delta}\left(\sum_{i=1}^n\sum_{k=1}^d\sum_{j=1}^{n_{ik}^{\operatorname{opt}}}(\widehat{c}_{ijk}^{~+}+\widehat{c}_{ijk}^{~-}+\widehat{p}_{ijk}^{~+}+\widehat{p}_{ijk}^{~-})\right)\\
&\leq \widetilde{\delta}+\delta (1+\kappa)  \mathfrak{B}  + {\delta} \mathfrak{B}<\varepsilon/2.
\end{aligned}
\end{equation}
\endgroup
Hence, we obtain that
\begin{align}
&\overline{D}^{\mathfrak{B},B}_{(\boldsymbol{\widetilde{K}},\widetilde{\boldsymbol{\pi}},\B{\widetilde{S}_{t_0}})}\left(\Phi_{\widetilde{\theta}}\right)-\overline{D}^{\mathfrak{B},B}_{(\boldsymbol{K},\boldsymbol{\pi},\B{S_{t_0}})}\left(\Phi_{\theta}\right) \notag \\
= &\inf_{\substack{( a, \B{c_{ijk}},\B{p_{ijk}}, \Delta_i^k): \\ \Psi^{(\boldsymbol{\widetilde{K}},\B{\widetilde{S}_{t_0}})}_{( a, \B{c_{ijk}},\B{p_{ijk}}, \Delta_i^k)} \geq \Phi_{\widetilde{\theta}}, \\ \Sigma\left({\B{{c}_{ijk}},\B{{p}_{ijk}},{\Delta}_0^k}\right) \leq \mathfrak{B}}} \C\left(\Psi^{(\boldsymbol{\widetilde{K}},\B{\widetilde{S}_{t_0}})}_{( a, \B{c_{ijk}},\B{p_{ijk}}, \Delta_i^k)},\widetilde{\boldsymbol{\pi}}\right) \notag  -\inf_{\substack{( a, \B{c_{ijk}},\B{p_{ijk}}, \Delta_i^k): \\ \Psi^{(\boldsymbol{K},\B{S_{t_0}})}_{( a, \B{c_{ijk}},\B{p_{ijk}}, \Delta_i^k)} \geq \Phi_{\theta}, \notag \\
\Sigma\left({\B{{c}_{ijk}},\B{{p}_{ijk}},{\Delta}_0^k}\right) \leq \mathfrak{B}}}\C\left(\Psi^{(\boldsymbol{K},\B{S_{t_0}})}_{( a, \B{c_{ijk}},\B{p_{ijk}}, \Delta_i^k)},\boldsymbol{\pi}\right) \notag \\
\leq &\C\left(\Psi^{(\boldsymbol{\widetilde{K}},\B{\widetilde{S}_{t_0}})}_{(\widetilde{\delta}+\delta(1+\kappa) \mathfrak{B}+\widehat{a},\B{\widehat{c}_{ijk}},\B{\widehat{p}_{ijk}},\widehat{\Delta}_i^k)},\widetilde{\boldsymbol{\pi}} \right) \notag -\C\left(\Psi^{(\boldsymbol{K},\B{S_{t_0}})}_{(\widehat{a},\B{\widehat{c}_{ijk}},\B{\widehat{p}_{ijk}},\widehat{\Delta}_i^k)} ,\boldsymbol{\pi} \right)+\varepsilon/2\notag
< \varepsilon, \notag
\end{align}
where the last two inequalities are consequences of \eqref{eq_f1df2deps}, \eqref{eq_greater_assumption}, \eqref{eq_c_epsilon_half}, and \eqref{eq_53a}.

If instead the inequality $\overline{D}^{\mathfrak{B},B}_{(\boldsymbol{K},\boldsymbol{\pi},\B{S_{t_0}})}\left(\Phi_{\theta}\right) \geq \overline{D}^{\mathfrak{B},B}_{(\boldsymbol{\widetilde{K}},\widetilde{\boldsymbol{\pi}},\B{\widetilde{S}_{t_0}})}\left(\Phi_{\widetilde{\theta}}\right)$ holds, then in this case we choose parameters
$\widehat{\widetilde{a}},(\B{\widehat{\widetilde{c}}_{ijk}})_{i,j,k}$, $(\B{\widehat{\widetilde{p}}_{ijk}})_{i,j,k}$, $(\widehat{\widetilde{\Delta}}_i^k)_{i,k}$ such that $\Psi^{(\boldsymbol{\widetilde{K}},\B{\widetilde{S}_{t_0}})}_{(\widehat{\widetilde{a}},\B{\widehat{\widetilde{c}}_{ijk}},\B{\widehat{\widetilde{p}}_{ijk}},\widehat{\widetilde{\Delta}}_i^k)}(\B{s}) \geq \Phi_{\widetilde{\theta}}(\B{s})$ for all $\B{s} \in [0,B]^{nd}$,  $\Sigma\left({\B{\widehat{\widetilde{c}}_{ijk}},\B{\widehat{\widetilde{p}}_{ijk}},\widehat{\widetilde{\Delta}}_i^k}\right) \leq \mathfrak{B}$, and such that 
\begin{equation*}
\C\left(\Psi^{(\boldsymbol{\widetilde{K}},\B{\widetilde{S}_{t_0}})}_{(\widehat{\widetilde{a}},\B{\widehat{\widetilde{c}}_{ijk}},\B{\widehat{\widetilde{p}}_{ijk}},\widehat{\widetilde{\Delta}}_i^k)},\widetilde{\boldsymbol{\pi}}\right)
\leq
\overline{D}^{\mathfrak{B},B}_{(\boldsymbol{\widetilde{K}},\widetilde{\boldsymbol{\pi}},\B{\widetilde{S}_{t_0}})}\left(\Phi_{\widetilde{\theta}}\right)+\varepsilon/2,
\end{equation*}
and then we repeat the following line of argumentation. This shows part (a).
\end{proof}

\begin{rem}
For all $\varepsilon>0$, $(\B{K},\B{\pi},\B{S}_{t_0},\theta) \in \mathbb{K}_1$, there exists some $\delta >0$ such that if
\eqref{eq_assumption_delta} holds for $(\boldsymbol{\widetilde{K}},\B{\widetilde{\pi}},\B{\widetilde{S}_{t_0}},\widetilde{\theta}) \in \mathbb{K}_1$, then for all strategies $\Psi^{(\boldsymbol{K},\B{S_{t_0}})}_{(a,\B{c_{ijk}},\B{p_{ijk}},\Delta_i^k)}$ satisfying $\Psi^{(\boldsymbol{K},\B{S_{t_0}})}_{(a,\B{c_{ijk}},\B{p_{ijk}},\Delta_i^k)} \geq \Phi_{\theta}$ on $[0,B]^{nd}$ and $\Sigma\left({\B{{c}_{ijk}},\B{{p}_{ijk}},{\Delta}_i^k}\right) \leq \mathfrak{B}$, there exists some $\widetilde{a} \in \R$ with 
\begin{equation}\label{eq_approximation_parameters}
|a-\widetilde{a}|<\varepsilon/2
\end{equation}
such that $\Psi^{(\boldsymbol{\widetilde{K}},\B{\widetilde{S}_{t_0}})}_{(\widetilde{a},\B{c_{ijk}},\B{p_{ijk}},\Delta_i^k)} \geq {\Phi}_{\widetilde{\theta}}$ on $[0,B]^{nd}$.
Indeed, let $\varepsilon>0$,  $(\B{K},\B{\pi},\B{S}_{t_0},\theta) \in \mathbb{K}_1$, and choose $\delta,\widetilde{\delta}>0$ analogue as in the proof of Theorem~\ref{thm_convergence}~(a), such that \eqref{ineq_lemma_delta}, \eqref{eq_max_ineq}, \eqref{eq_f1df2deps}, and \eqref{eq_assumption_delta} hold.
Then according to \eqref{eq_proof_f1d} we see that the strategy $\Psi^{(\boldsymbol{\widetilde{K}},\B{\widetilde{S}_{t_0}})}_{(\widetilde{\delta}+\delta (1+\kappa)  \mathfrak{B} +a,\B{c_{ijk}},\B{p_{ijk}},\Delta_i^k)}$ fulfils \eqref{eq_approximation_parameters}.
\end{rem}
\begin{proof}[Proof of Theorem~\ref{thm_convergence}~(b)]
According to Theorem~\ref{thm_convergence}~(a), the map 
\[
(\boldsymbol{K},\boldsymbol{\pi},\B{S_{t_0}},\theta)\mapsto \left(\overline{D}^{\mathfrak{B},B}_{(\boldsymbol{K},\boldsymbol{\pi},\B{S_{t_0}})}\left(\Phi_{\theta}\right),\underline{D}^{\mathfrak{B},B}_{(\boldsymbol{K},\boldsymbol{\pi},\B{S_{t_0}})}\left(\Phi_{\theta}\right) \right)
\]
is an element of $C({\mathbb{K}_1},\R^2)$.
Hence, we find according to Proposition~\ref{lem_universal} a neural network $\mathcal{N}_1\in \mathfrak{N}_{N_{\operatorname{input}},2}$ such that \eqref{eq_neural_net_approx} holds on ${\mathbb{K}_1}$.
\end{proof}
\begin{proof}[Proof of Theorem~\ref{thm_convergence}~(c)]
Consider a sequence $\left(\boldsymbol{K}^{(n)},\B{\pi}^{(n)},\B{S_{t_0}}^{(n)},\theta^{(n)}\right)_{n \in \N} \subset \mathbb{K}_2$ with
\begingroup\makeatletter\def\f@size{8}\check@mathfonts
\def\maketag@@@#1{\hbox{\m@th\normalsize\normalfont#1}}
\[
\lim_{n\rightarrow \infty} d_{N_{\operatorname{input}}}\left((\boldsymbol{K}^{(n)},\B{\pi}^{(n)},\B{S_{t_0}}^{(n)},\theta^{(n)}),(\boldsymbol{K},\B{\pi},\B{S_{t_0}},\theta)\right)=0
\]
\endgroup
for some $(\boldsymbol{K},\B{\pi},\B{S_{t_0}},\theta) \in \mathbb{K}_2$. Since by assumption $\mathfrak{B}<\infty$ and by \eqref{eq_boundedness_assumption}, the sequence
\begingroup\makeatletter\def\f@size{7}\check@mathfonts
\def\maketag@@@#1{\hbox{\m@th\normalsize\normalfont#1}}
\[
x^{(n)}:=\left(a^*,(c^*_{1jk})_{j,k},(p^*_{1jk})_{j,k},({\Delta_0^k}^*)_k\right)(\boldsymbol{K}^{(n)},\B{\pi}^{(n)},\B{S_{t_0}}^{(n)},\theta^{(n)}),
\]
\endgroup
$n\in \N$, is bounded. Thus, there exists at least one accumulation point $x \in \R^{1+4M+d}$. Hence, we can find a subsequence $(x^{(n_k)})_{k\in \N}$ with $\lim_{k \rightarrow \infty} d_{1+4M+d} \left(x^{(n_k)},x\right)=0$. Then, we obtain by the continuity of $\overline{D}^{\mathfrak{B},B}$, shown in Theorem~\ref{thm_convergence}~(a), and by the continuity of the cost function $\C$, defined in \eqref{eq_definition_cost_functional}, w.r.t.\,all its arguments, that
\begin{align*}
\overline{D}^{\mathfrak{B},B}_{(\boldsymbol{K},\boldsymbol{\pi},\B{S_{t_0}})}\left(\Phi_{\theta}\right)=&\lim_{k \rightarrow \infty} \overline{D}^{\mathfrak{B},B}_{\left(\boldsymbol{K}^{(n_k)},\boldsymbol{\pi}^{(n_k)},\B{S_{t_0}}^{(n_k)}\right)}\left(\Phi_{\theta^{(n_k)}}\right)
=\lim_{k \rightarrow \infty} \C\left(\Psi^{(\boldsymbol{K}^{(n_k)},\B{S_{t_0}}^{(n_k)})}_{x^{(n_k)}},\boldsymbol{\pi}^{(n_k)}\right)
=\C\left(\Psi^{(\boldsymbol{K},\B{S_{t_0}})}_{x},\boldsymbol{\pi}\right).
\end{align*}
Using the continuity of $\Psi$ w.r.t.\,its parameters and the continuity of $\theta \mapsto \Phi_\theta$ as in \eqref{eq_defn_theta_map}, we also obtain that $\Psi^{(\boldsymbol{K},\B{S_{t_0}})}_{x} \geq \Phi_\theta$ on $[0,B]^{nd}$. Moreover, $x:=\left(a,(c_{1jk})_{j,k},(p_{1jk})_{j,k},({\Delta_0^k})_k\right)\in \R^{1+4M+d}$ satisfies $\Sigma(\B{c_{1jk}},\B{p_{1jk}},\Delta_0^k) \leq \mathfrak{B}$. Thus, $x$ is indeed a minimal super-replication strategy of $\Phi_\theta$ for parameters $(\boldsymbol{K}, \B{\pi},\B{S_{t_0}},\theta)$. So we have shown that any accumulation point $x$ is a minimal super-replication strategy of $\Phi_\theta$ for parameters $(\boldsymbol{K}, \B{\pi}, \B{S_{t_0}},\theta)$. Since we assumed that the minimizer is unique, the accumulation point $x$ is unique and is necessarily the limit of the sequence $(x^{(n)})_{n\in \N}$. Therefore, we have shown that
\begingroup\makeatletter\def\f@size{8}\check@mathfonts
\def\maketag@@@#1{\hbox{\m@th\normalsize\normalfont#1}}
\begin{align*}
\lim_{n\rightarrow\infty}&\left(a^*,(c^*_{1jk})_{j,k},(p^*_{1jk})_{j,k},({\Delta_0^k}^*)_k\right)(\boldsymbol{K}^{(n)}, \B{\pi}^{(n)},\B{S_{t_0}}^{(n)},\theta^{(n)})
=
\left(a^*,(c^*_{1jk})_{j,k},(p^*_{1jk})_{j,k},({\Delta_0^k}^*)_k\right)\left(\boldsymbol{K}, \B{\pi}, \B{S_{t_0}}, \theta\right).
\end{align*}
\endgroup
\end{proof}
\begin{proof}[Proof of Theorem~\ref{thm_convergence}~(d)]~\\
According to Theorem~\ref{thm_convergence}~(b), the map
\begin{align*}
(\boldsymbol{K},\boldsymbol{\pi},\B{S_{t_0}},\theta) \mapsto \bigg(&a^*,(c^*_{1jk})_{j,k},(p^*_{1jk})_{j,k},({\Delta_0^k}^*)_k\bigg)\big(\boldsymbol{K},\B{\pi}, \B{S_{t_0}}, \theta\big),
\end{align*}
is an element of $C({\mathbb{K}_2},\R^{1+4M+d})$.
Thus, by Proposition~\ref{lem_universal}, there exists some $\mathcal{N}_2 \in \mathfrak{N}_{N_{\operatorname{input}},1+4M+d}$ such that \eqref{eq_neural_net_approx_strats} holds on ${\mathbb{K}_2}$.
\end{proof}
\begin{proof}[Proof of Theorem~\ref{lem_approx_mot}]
Let $\varepsilon>0$, $N \in \N$, and pick some compact set $\mathbb{K}\subset \R_+^N$.
We observe that the map 
\begin{align*}
\mathcal{D}^{(N)}:\left(\R_+^N,d_N\right) &\rightarrow \left(\mathcal{P}_1(\R_+), \operatorname{W}\right)\\
(x_1,\dots,x_N)&\mapsto \frac{1}{N}\sum_{i=1}^N \delta_{x_i}
\end{align*}
is continuous.
Indeed, if $$\R_+^N \ni \B{x}^m =(x_1^m,\dots,x_N^m)\rightarrow \B{x}=(x_1,\dots,x_N)\in \R_+^N$$ for $m \rightarrow \infty$, then we can consider for all $m \in \N$ the coupling $\pi^m:= \frac{1}{N} \sum_{i=1}^N\delta_{(x_i,x_i^m)}$ and obtain that
\begin{equation}\label{eq_wasserstein_0}
\begin{aligned}
\operatorname{W}\left(\mathcal{D}^{(N)}(\B{x}),\mathcal{D}^{(N)}(\B{x}^m)\right)&\leq  \int_{\R_+^2} | u-v |\D\pi^m (u,v) = \frac{1}{N} \sum_{i=1}^N|x_i-x_i^m|\rightarrow 0\text{ for } m \rightarrow \infty.
\end{aligned}
\end{equation}
Since by assumption $\Phi$ is upper semi-continuous with $\sup_{x_1,x_2 \in \R_+}\frac{|\Phi(x_1,x_2)|}{1+x_1+x_2} < \infty $, we can apply \cite[Theorem 2.9.]{wiesel2019continuity} which ensures that for any $(\mu_1,\mu_2) \in \mathcal{P}_1(\R_+)\times\mathcal{P}_1(\R_+)$ with $\mu_1 \preceq \mu_2$ and $(\mu_1^m,\mu_2^m) \in \mathcal{P}_1(\R_+)\times \mathcal{P}_1(\R_+)$ with $\mu_1^m \preceq \mu_2^m$, $m \in \N$, satisfying  $\operatorname{W}(\mu_1^m,\mu_1) \rightarrow 0$, $\operatorname{W}(\mu_2^m,\mu_2) \rightarrow 0$ for $m \rightarrow \infty$, it holds
\begin{equation}\label{eq_continuity_mot}
\lim_{m \rightarrow \infty}\left|\sup_{\Q \in \mathcal{M}(\mu_1^m,~\mu_2^m)} \E_\Q[\Phi]-\sup_{\Q \in \mathcal{M}(\mu_1,~\mu_2)} \E_\Q[\Phi] \right|=0.
\end{equation}
Since for any set of measures $\mathcal{M}\subset \mathcal{P}(\R_+^2)$ we have  $\inf_{\Q \in \mathcal{M}}\E_\Q[\Phi]=-\sup_{\Q \in \mathcal{M}}\E_\Q[-\Phi]$, and since $-\Phi$ is also upper semi-continuous with $\sup_{x_1,x_2 \in \R_+}\frac{|-\Phi(x_1,x_2)|}{1+x_1+x_2} < \infty $, we can apply the same arguments to see that
\begin{equation}\label{eq_continuity_mot_lower_bound}
\lim_{m \rightarrow \infty}\left|\inf_{\Q \in \mathcal{M}(\mu_1^m,~\mu_2^m)} \E_\Q[\Phi]-\inf_{\Q \in \mathcal{M}(\mu_1,~\mu_2)} \E_\Q[\Phi] \right|=0.
\end{equation}
Define the closed set\footnote{We refer to \cite[Definition 2.45]{baker2012martingales} and \cite[Lemma 2.49]{baker2012martingales} for a characterization of $\B{x},\B{y} \in \R_+^N$ to satisfy \newline $\mathcal{D}^{(N)}(\B{x}) \preceq \mathcal{D}^{(N)}(\B{y})$.} $\mathfrak{C}^{(N)}:=\left\{(\B{x},\B{y}) \in \R_+^N\times \R_+^N~:~\mathcal{D}^{(N)}(\B{x}) \preceq \mathcal{D}^{(N)}(\B{y})\right\}$.
Then, we obtain by \eqref{eq_wasserstein_0}, \eqref{eq_continuity_mot}, and \eqref{eq_continuity_mot_lower_bound} the continuity of 
\begin{equation}
\begin{aligned}
(\B{x},\B{y}) \mapsto \Bigg(&\inf_{\Q \in \mathcal{M}\left(\mathcal{D}^{(N)}(\B{x}),~\mathcal{D}^{(N)}(\B{y})\right)} \E_\Q[\Phi],\sup_{\Q \in \mathcal{M}\left(\mathcal{D}^{(N)}(\B{x}),~\mathcal{D}^{(N)}(\B{y})\right)} \E_\Q[\Phi]\Bigg) \qquad \text{on } \mathfrak{C}^{(N)}.
\end{aligned}
\end{equation}
Hence, the universal approximation theorem from Proposition~\ref{lem_universal} guarantees the existence of a neural network ${\mathcal{N}} \in \mathfrak{N}_{2N,2}$ such that 
\begingroup\makeatletter\def\f@size{9.2}\check@mathfonts
\def\maketag@@@#1{\hbox{\m@th\normalsize\normalfont#1}} 
\begin{equation}\label{eq_eps_2_ineq_1}
\begin{aligned}
\sup_{\B{x},\B{y} \in \mathbb{K} \cap \mathfrak{C}^{(N)}}\Bigg\|&\mathcal{N}(\B{x},\B{y})-\Bigg(\inf_{\Q \in \mathcal{M}\left(\mathcal{D}^{(N)}(\B{x}),~\mathcal{D}^{(N)}(\B{y})\right)} \E_\Q[\Phi],\sup_{\Q \in \mathcal{M}\left(\mathcal{D}^{(N)}(\B{x}),~\mathcal{D}^{(N)}(\B{y})\right)} \E_\Q[\Phi]\Bigg)\Bigg\|_2 < \varepsilon/2.
\end{aligned}
\end{equation}
\endgroup
Now let $(\mu_1,\mu_2) \in \mathcal{P}_1(\R_+)\times\mathcal{P}_1(\R_+)$ with $\mu_1 \preceq \mu_2$. Then \cite[Theorem 2.4.11.]{baker2012martingales} ensures that 
\begin{equation}\label{eq_convex_order_u}
\mathcal{U}^{(N)}(\mu_1)\preceq \mathcal{U}^{(N)}(\mu_2).
\end{equation}
This and the continuity of the two-marginal MOT problem with respect to its marginals, as stated in \eqref{eq_continuity_mot} and \eqref{eq_continuity_mot_lower_bound}, implies that there exists some $\delta>0$ such that if $\operatorname{W}(\mathcal{U}^{(N)}(\mu_1),\mu_1)<\delta$, $\operatorname{W}(\mathcal{U}^{(N)}(\mu_2),\mu_2)<\delta$, then
%%%%%%%
\begingroup\makeatletter\def\f@size{9}\check@mathfonts
\def\maketag@@@#1{\hbox{\m@th\normalsize\normalfont#1}} 
\begin{equation}\label{eq_eps_2_ineq_2}
\begin{aligned}
\bigg\|\bigg(\inf_{\Q \in \mathcal{M}(\mu_1,\mu_2)}\E_{\Q}[\Phi],&\sup_{\Q \in \mathcal{M}(\mu_1,\mu_2)}\E_{\Q}[\Phi]\bigg)
-\Bigg(\inf_{\Q \in \mathcal{M}(\mathcal{U}^{(N)}(\mu_1) ,\mathcal{U}^{(N)}(\mu_2) )}\E_{\Q}[\Phi],\sup_{\Q \in \mathcal{M}(\mathcal{U}^{(N)}(\mu_1) ,\mathcal{U}^{(N)}(\mu_2) )}\E_{\Q}[\Phi]\Bigg)\bigg\|_2<\varepsilon/2.
\end{aligned}
\end{equation}
\endgroup
%%%%%%%%%%%
By definition of the map $\mathcal{D}^{(N)}$ it holds that $\mathcal{D}^{(N)}\left(\B{x}^{(N)}(\mu_i)\right)=\mathcal{U}^{(N)}(\mu_i)$ for $i=1,2$. In particular, we have by \eqref{eq_convex_order_u} that $\left(\B{x}^{(N)}(\mu_1),\B{x}^{(N)}(\mu_2)\right)\in \mathfrak{C}^{(N)}$. Thus, the triangle inequality and \eqref{eq_eps_2_ineq_1} combined with \eqref{eq_eps_2_ineq_2} implies that if $\operatorname{W}(\mathcal{U}^{(N)}(\mu_1),\mu_1)<\delta$, $\operatorname{W}(\mathcal{U}^{(N)}(\mu_2),\mu_2)<\delta$, and $\B{x}^{(N)}(\mu_1),\B{x}^{(N)}(\mu_2) \in \mathbb{K}$, then
%%%%%%%
\begingroup\makeatletter\def\f@size{8}\check@mathfonts
\def\maketag@@@#1{\hbox{\m@th\normalsize\normalfont#1}} 
\begin{align*}
&\Bigg\|\mathcal{N}\left(\B{x}^{(N)}(\mu_1),\B{x}^{(N)}(\mu_2)\right)-\left(\inf_{\Q \in \mathcal{M}(\mu_1,\mu_2)}\E_\Q[\Phi],\sup_{\Q \in \mathcal{M}(\mu_1,\mu_2)}\E_\Q[\Phi]\right)\Bigg\|_2\\
\leq &\Bigg\|\mathcal{N}\left(\B{x}^{(N)}(\mu_1),\B{x}^{(N)}(\mu_2)\right)-\bigg(\inf_{\Q \in \mathcal{M}(\mathcal{U}^{(N)}(\mu_1) ,\mathcal{U}^{(N)}(\mu_2) )}\E_\Q[\Phi],\sup_{\Q \in \mathcal{M}(\mathcal{U}^{(N)}(\mu_1) ,\mathcal{U}^{(N)}(\mu_2) )}\E_\Q[\Phi]\bigg)\Bigg\|_2\\
&+\Bigg\|\bigg(\inf_{\Q \in \mathcal{M}(\mathcal{U}^{(N)}(\mu_1) ,\mathcal{U}^{(N)}(\mu_2) )}\E_\Q[\Phi],\sup_{\Q \in \mathcal{M}(\mathcal{U}^{(N)}(\mu_1) ,\mathcal{U}^{(N)}(\mu_2) )}\E_\Q[\Phi]\bigg)
-\left(\inf_{\Q \in \mathcal{M}(\mu_1,\mu_2)}\E_\Q[\Phi],\sup_{\Q \in \mathcal{M}(\mu_1,\mu_2)}\E_\Q[\Phi]\right)\Bigg\|_2\\
< &\varepsilon,
\end{align*}
\endgroup
%%%%%%%%%
which shows the assertion.
\end{proof}
%%%%%%%%%%%%%%%%%%%%%%%%%%%%%%%%%%%%%%%%%%%%%%%%%%%%%%%%
 \section*{Acknowledgment}
Financial support by the Nanyang Assistant Professorship Grant (NAP Grant) \emph{Machine Learning based Algorithms in Finance and Insurance} is gratefully acknowledged. 
	\noindent

}

%%%%%%%%%%%%%%%%%%%%%%%%%%%%%%%%%%%%
%\section{Conclusion}
\bibliographystyle{plain}
%\bibliography{biblio}

%%%%%%%%%%%%%%%%%%%%%%%%%%%%%%%%%%%%%

      \end{document}